\documentclass[[twocolumn,
preprintnumbers,superscriptaddress
,footinbib
 aps,
 prresearch,
 reprint,
 superscriptaddress,
 amsmath,
 amssymb
]{revtex4-2}

\usepackage{graphicx}
\usepackage{dcolumn}
\usepackage{bm}
\usepackage{amsmath, amsfonts, amsthm, amssymb, type1cm, braket}
\usepackage{url}
\usepackage{appendix}
\usepackage{tikz}
\usepackage{caption}
\usepackage{mathtools}
\usepackage{subcaption} 
\usepackage[normalem]{ulem}
\usepackage{cancel}

\newcommand{\Tr}{\text{Tr}}

\newcommand{\Prob}{\text{Prob}}
\newcommand{\phys}{\text{phys}}

\newcommand{\tot}{\text{tot}}
\newcommand{\est}{\text{est}}
\newcommand{\CNOT}{\text{CNOT}}
\newcommand{\shot}{\text{shot}}

\newcommand{\chain}{\text{chain}}
\newcommand{\Var}{\text{Var}}
\newcommand{\Cov}{\text{Cov}}

\newcommand{\Gibbs}{\text{Gibbs}}
\newcommand{\Exp}{\mathbb{E}}
\newcommand{\Obs}{\mathcal{O}}

\newcommand{\Hilbert}{\mathcal{H}}
\newcommand{\Pauli}{\mathcal{P}}
\newcommand{\Clifford}{\mathcal{C}}
\newcommand{\Stab}{\mathcal{S}}
\newcommand{\Unitary}{\mathcal{U}}

\newcommand{\Complex}{\mathbb{C}}

\newcommand{\beq}{\begin{eqnarray}}
\newcommand{\eeq}{\end{eqnarray}}
\newcommand{\e}{\mathrm{e}}

\newif\ifmarked
\markedfalse    

\ifmarked
  \newcommand{\rev}[1]{\textcolor{red}{#1}}
  \newcommand{\revdel}[1]{\textcolor{red}{\cancel{#1}}}
\else
  \newcommand{\rev}[1]{#1}
  \newcommand{\revdel}[1]{}
\fi

\newtheorem{theorem}{Theorem}[section]
\newtheorem{lemma}[theorem]{Lemma}

\newtheorem{corollary}[theorem]{Corollary}
\theoremstyle{definition}
\newtheorem{definition}[theorem]{Definition}

\usepackage{hyperref}

\begin{document}

\preprint{YITP-26-25}

\title{Gauge-invariant QMETTS with mutually unbiased physical bases for $\mathbb{Z}_2$ lattice gauge theories at finite temperature and density}

\author{Reita Maeno}
\email{maeno-reita916@g.ecc.u-tokyo.ac.jp}
\affiliation{Department of Physics, Graduate School of Science, The University of Tokyo, Tokyo 113-0033, Japan}
\affiliation{Yukawa Institute for Theoretical Physics, Kyoto University, Kitashirakawa Oiwakecho, Sakyo-ku, Kyoto 606-8502, Japan}

\date{\today}

\begin{abstract}
In quantum computations of gauge theories at finite temperature and finite density, enforcing Gauss’s law for all states contributing to the thermal ensemble is a nontrivial challenge.
In this work, we adapt the Quantum Minimally Entangled Typical Thermal States (QMETTS) algorithm to $\mathbb{Z}_2$ gauge-constrained systems and propose a method for computing finite-temperature and finite-density expectation values without eliminating redundant gauge-field degrees of freedom.
In QMETTS, the thermal ensemble is sampled via a Markov chain of pure states generated by imaginary-time evolution and projective measurements.
To preserve gauge invariance while maintaining efficient sampling, we introduce measurement bases that are gauge invariant and mutually unbiased within the physical subspace.
We show that such measurement bases can be constructed efficiently for $\mathbb{Z}_2$ lattice gauge theories in arbitrary spatial dimensions and arbitrary boundary conditions by exploiting the correspondence between $\mathbb{Z}_2$ lattice gauge theories and the stabilizer formalism.
Furthermore, since expectation-value estimation on quantum hardware is inherently affected by shot noise, we explicitly incorporate shot noise into the analysis. 
By formulating a finite-shot version of QMETTS, we show that the resulting estimator remains unbiased and that using one observable-measurement shot per sampled state is nearly optimal in terms of variance when the total number of circuit executions used to generate the Markov chain and measure observables is fixed.
This result indicates that it is often more efficient to generate more QMETTS samples than to accurately estimate the expectation value for each individual pure state.
We validate the proposed method numerically in a $(1+1)$-dimensional $\mathbb{Z}_2$ lattice gauge theory coupled to staggered fermions.
Our results provide a gauge-invariant sampling framework for finite-temperature and finite-density quantum algorithms for lattice gauge theories.
\end{abstract}

\maketitle


\section{Introduction}
Understanding the properties of equilibrium systems at finite temperature and finite density in gauge theories is a central issue in modern theoretical physics~\cite{Stephanov:2006zvm, Fukushima:2010bq, shuryak1980quantum}.
Ab initio numerical studies of gauge theories are among the most powerful tools for investigating such phenomena, and lattice gauge theory provides a natural framework for this purpose through lattice regularization~\cite{Wilson:1974sk, Kogut:1974ag}.
Conventional Monte Carlo studies of lattice gauge theories based on the Euclidean Lagrangian formalism have achieved remarkable success in reproducing a wide range of nonperturbative phenomena~\cite{durr2008ab}.
However, at finite density, these approaches suffer from the notorious sign problem, which severely limits their applicability~\cite{Nagata:2021ugx}.

Recently, quantum computing and quantum simulation have attracted considerable attention as promising routes to overcome this limitation.
For finite-temperature and finite-density equilibrium calculations, however, one must access thermal ensembles rather than a single pure state.
Although various quantum algorithms have been proposed for this purpose, preparing thermal states on quantum devices remains challenging~\cite{Chen:2023cuc, Chen:2023zpu, Consiglio:2023mev, Sewell:2022nwi}.
One attractive alternative is the Quantum Minimally Entangled Typical Thermal States (QMETTS) algorithm~\cite{Motta:2019yya}.
QMETTS estimates finite-temperature and finite-density observables by sampling an ensemble of pure states through a Markov chain, avoiding directly preparing a full thermal mixed state, which can require substantial quantum resources.

In addition to these general challenges of finite-temperature calculations, gauge theories are subject to constraints imposed by Gauss's law.
Physical thermal states are supported on the physical Hilbert space, namely the subspace spanned by states satisfying the gauge constraints.
This constraint makes the design of quantum algorithms for finite-temperature gauge theories nontrivial.
One common strategy is to explicitly solve the gauge constraints and work directly within the physical Hilbert space\rev{~\cite{Chen:2024oao,Tomiya:2022chr, Turro:2024pxu, Turro:2025sec, Than:2024zaj}}.
While this removes unphysical degrees of freedom, it can make the Hamiltonian more complicated.
In addition, retaining the redundant gauge-field degrees of freedom can itself be advantageous, since gauge constraints provide useful diagnostics and can be exploited for error mitigation~\cite{Schmale:2024ebh, Ballini:2024qmr, Rajput:2021trn, Carena:2024dzu, Spagnoli:2024mib}.
Therefore, it is important to develop methods that evaluate finite-temperature expectation values without explicitly eliminating the redundant Hilbert-space structure\rev{~\cite{Davoudi:2022uzo, Fromm:2023npm, Ballini:2023ljs}}.

In this work, we develop a QMETTS-based framework for computing gauge-invariant expectation values while retaining the redundant Hilbert-space structure.
QMETTS consists mainly of two steps: imaginary-time evolution and projective measurements that generate the next pure state in the Markov chain.
The imaginary-time evolution step preserves gauge symmetry as long as the Hamiltonian commutes with the Gauss's law operators.
The projective-measurement step, on the other hand, is more subtle.
In conventional many-body systems, efficient QMETTS sampling is often achieved by alternating mutually unbiased bases, such as the $Z$- and $X$-bases (see Eq.~(\ref{eq:MUB}) for the definition of mutually unbiasedness), whereas in gauge theories these standard bases generally violate the Gauss's law constraints.
Therefore, realizing QMETTS for gauge theories requires measurement bases that satisfy two nontrivial requirements simultaneously: they must preserve gauge symmetry and retain mutual unbiasedness needed for efficient sampling.

To address this issue, we introduce mutually unbiased physical bases (MUPB) for $\mathbb{Z}_2$ gauge theories, which satisfy the above two conditions.
By exploiting the correspondence between Gauss's law constraints and the stabilizer formalism, which is well developed in the context of quantum error correction~\cite{nielsen2010quantum, Gottesman:1997zz}, we present a general construction of such bases in arbitrary dimensions and with arbitrary boundary conditions.
The resulting circuits require $\mathcal{O}(N_q^2)$ gates and have depth $\mathcal{O}(\log N_q)$, where $N_q$ is the number of qubits.
Although related tensor-network studies have developed symmetry-adapted measurement bases for METTS, mainly for condensed-matter systems with global symmetries~\cite{Binder:2017bxd, goto2019quasiexact}, extending these ideas to local gauge constraints and implementing the resulting measurements as scalable, circuit-efficient quantum procedures is nontrivial.
Our construction is tailored to the quantum-computing setting and provides explicit quantum circuits for gauge-preserving mutually unbiased measurements.
A more detailed comparison with these tensor-network approaches is given in Sec.~\ref{subsec:contribution}.

In addition to constructing MUPBs, we analyze the effect of shot noise on expectation-value estimation in QMETTS.
Although shot noise is unavoidable on quantum hardware, its role in QMETTS has not been fully clarified.
We formulate a finite-shot version of QMETTS and analyze how finite measurement statistics affect the mean and variance of the estimator.
In this analysis, we show that the estimator remains unbiased under the stationary Markov-chain.
We then derive the optimal number of observable-measurement shots for each sampled pure state under a fixed total number of circuit executions used to generate the Markov chain and measure observables.
While this optimal shot number is generally difficult to know a priori, we show that the single-shot strategy is near-optimal: its variance is at most twice the optimum.

As a numerical demonstration, we apply our algorithm to the $(1+1)$-dimensional $\mathbb{Z}_2$ lattice gauge theory coupled to staggered fermions.
We investigate the performance of the proposed method and verify gauge-invariant sampling, improved mixing, and accurate finite-temperature and finite-density expectation values.
\rev{Furthermore, to assess the practical performance of our framework, we investigate our algorithm in combination with the Quantum Imaginary Time Evolution (QITE) method proposed by Motta \textit{et al.}~\cite{Motta:2019yya}, which provides one approach to implementing imaginary-time evolution on quantum computers.}

The rest of this paper is organized as follows.
Section~\ref{sec:LGTsetup} introduces the $(1+1)$-dimensional $\mathbb{Z}_2$ lattice gauge theory as a concrete example.
Section~\ref{sec:QMETTS} briefly reviews the QMETTS algorithm.
In Sec.~\ref{sec:MUPB}, we present our approach to evaluating gauge-invariant expectation values at finite temperature and finite density within the QMETTS framework.
In Sec.~\ref{sec:shot_analysis}, we study statistical errors in the estimation of expectation values and optimize the number of shots.
Based on Secs.~\ref{sec:MUPB} and \ref{sec:shot_analysis}, Sec.~\ref{sec:numerical} presents numerical demonstrations for the $\mathbb{Z}_2$ lattice gauge theory at finite temperature and finite density.
Finally, Sec.~\ref{sec:conclusion} summarizes our results and discusses future directions.

\section{Model} \label{sec:LGTsetup}
For concreteness, we first focus on (1+1)-dimensional $\mathbb{Z}_2$ lattice gauge theory (LGT) coupled to staggered fermions, which is the simplest discrete realization of the Schwinger model~\cite{Horn:1979fy}.
Despite its minimal local Hilbert space consisting of two-level systems,
this model exhibits nontrivial properties such as confinement and has been widely used
as a benchmark for quantum simulations of lattice gauge theories~\cite{Mildenberger:2022jqr, Hayata:2024smx}.
In the staggered formulation, staggered fermions at site $2n\rev{-1}$ and $2n\revdel{{+1}}$, $\chi_{2n\rev{-1}}$ and $\chi_{2n\revdel{{+1}}}$, represent the two components of a Dirac spinor at site $n$, $\psi_n$.
Explicitly, a Dirac fermion $\psi_n$ is represented by
\begin{align}
    \frac{\psi_n}{\sqrt{a}} =
    \begin{pmatrix}
        \chi_{2{n\rev{-1}}} \\
        \chi_{2n\revdel{{+1}}}
    \end{pmatrix}, \nonumber
\end{align}
where $a$ denotes the lattice spacing.
In this formulation, staggered fermions reside on lattice sites and gauge fields live on the links between neighboring sites.
An annihilation (creation) operator of the staggered fermion $\chi_n (\chi_n^{\dagger})$  satisfies the canonical anticommutation relation $
    \{\chi_n, \chi_m^{\dagger}\} = \delta_{n,m}\label{eq:anti-commu}
$.
The Hamiltonian is given by
\begin{align}
    H &= \frac{1}{2a} \sum_{n=1}^{L_{\text{KS}}-1} \left[\chi_n^{\dagger} \sigma^{Z}_{n, n+1} \chi_{n+1} + \text{h.c.}\right]  \nonumber \\
    & + a g^2 \sum_{n=1}^{L_{\text{KS}}-1} \left[ 1 - \sigma^{X}_{n, n+1} \right] + {m} \sum_{n=1}^{L_{\text{KS}}} (-1)^{n\rev{+1}} \chi_n^{\dagger} \chi_n, \label{eq:original_Ham}
\end{align}
where $g$ and $m$ are the gauge coupling and the fermion mass, respectively.
$\sigma^Z_{n,n+1}$ and $\sigma^X_{n,n+1}$ are Pauli operators in the $\mathbb{Z}_2$ link and electric field operators, respectively.
As noted above, the number of Dirac-fermion sites, $L_{D}$, is half the number of staggered-fermion sites, $L_{\text{KS}}$, i.e., $L_{D}=L_{\text{KS}}/2$.

After mapping fermionic degrees of freedom into spin operators via Jordan--Wigner transformation~\cite{Jordan:1928wi}, the qubit description of the above Hamiltonian (\ref{eq:original_Ham}) is given in the following form;
\begin{align}
    H &= -\frac{1}{4a} \sum_{n=1}^{L_{\text{KS}}-1} \left[X_{n} \sigma^{Z}_{n, n+1} X_{n+1} + Y_{n} \sigma^{Z}_{n, n+1} Y_{n+1}\right]  \nonumber \\
    & + a g^2 \sum_{n=1}^{L_{\text{KS}}-1} \left[ 1 - \sigma^{X}_{n, n+1} \right] + \frac{m}{2} \sum_{n=1}^{L_{\text{KS}}} (-1)^{n\rev{+1}} Z_n, \nonumber
\end{align}
where $X_n, Y_n$, and $Z_n$ are usual Pauli matrices acting on the qubit at site $n$.
In this representation, as shown in Fig.~\ref{fig:setup}, the qubit basis is defined as follows.
On odd sites, which correspond to particle states, $\ket{0}$ ($\ket{1}$) corresponds to an occupied (empty) fermionic state,
whereas on even (antiparticle) sites, $\ket{1}$ ($\ket{0}$) represents an occupied (empty) state.
For gauge links, $\ket{+}$ ($\ket{-}$) represents right (left) directional flux.
We impose open boundary conditions, with the left boundary fixed by a virtual link in the $\ket{+}$ state.
\begin{figure}
    \centering
    \includegraphics[width=1.0\linewidth]{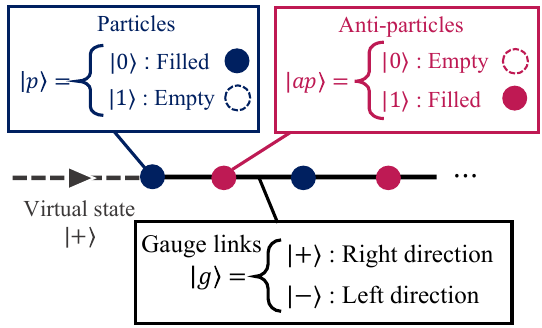}
    \caption{Schematic setup for $(1+1)$-dimensional $\mathbb{Z}_2$ LGT coupled to staggered fermions.
    Occupied sites correspond to $\ket{0}$ ($\ket{1}$) on particle (antiparticle) sites,
    while empty sites correspond to $\ket{1}$ ($\ket{0}$), respectively.
    For gauge links, $\ket{+}$ represents right directional flux and $\ket{-}$ represents left.
    We impose open boundary conditions, with the left boundary fixed by a virtual link in the $\ket{+}$ state.
    } 
    \label{fig:setup}
\end{figure}

In addition, gauge theories are subject to local constraints associated with gauge symmetry, known as Gauss’s law constraints.
In our setup, Gauss's law operators are represented by
\begin{align}
    G_n 
    &= (-1)^{n} \sigma^X_{n-1, n} Z_{n} \sigma^X_{n, n+1}  ~~ (1\leq n \leq L_{\text{KS}}-1). \label{eq:gauss' law op}
\end{align}
Due to the boundary condition, we fix $\sigma_{0, 1}^X= 1$.
Their eigenvalues are $g_n = \pm 1$.
They commute with each other, $[G_n, G_m]=0, ~\forall n, m$ and with the Hamiltonian $[H, G_n]=0, ~ \forall n$.
In this work, we choose a sector without background, called physical sector, defined by $g_n = +1, ~ \forall n$.
Once we restrict to the physical sector, the states contributing to expectation values are only physical states $\ket{\psi}$ satisfying $G_n \ket{\psi} = +\ket{\psi}, ~ \forall n$.
We denote the subspace spanned by physical states by $\Hilbert_{\phys}$ and total Hilbert space by $\Hilbert_{\tot}$.
Enforcing this constraint poses a nontrivial challenge for the computation of gauge theories.

To study thermal-equilibrium properties of the theory, we consider the grand-canonical expectation values of an observable $\Obs$ under inverse temperature $\beta=1/T$ and chemical potential $\mu$:
\begin{align}
    \braket{\Obs}_{\beta, \mu} := \Tr_{\Hilbert_{\phys}} \left[ \Obs \frac{e^{-\beta(H-\mu N)}}{Z}\right], \label{eq:trace}
\end{align}
where \rev{$Z = \Tr_{\Hilbert_{\phys}}[e^{-\beta (H - \mu N)}]$ is the partition function and} $N$ is the number operator given by
\begin{align}
    N := \sum_{n=1}^{L_{\text{KS}}} \chi_n^\dagger \chi_n -\frac{L_{\text{KS}}}{2} I= \frac{1}{2} \sum_{n=1}^{L_{\text{KS}}}Z_n. \label{eq:number op}
\end{align}
Here we have dropped the additive constant that appears in the Jordan--Wigner
mapping of the total staggered occupation number, since it only shifts
\(H-\mu N\) by a constant and therefore does not affect normalized thermal
expectation values.
Note that the trace $\Tr_{\Hilbert_{\phys}}[\cdot]$ in Eq.~(\ref{eq:trace}) is taken over the physical subspace $\Hilbert_{\phys}$.
Since the number operator $N$ in Eq.~(\ref{eq:number op}) commutes with all Gauss's law operators $G_n$ in Eq.~(\ref{eq:gauss' law op}), the grand-canonical ensemble is compatible with gauge invariance.

For $\mathbb{Z}_2$ LGTs in a general setup, as in the $(1+1)$-dimensional case, the Gauss's law operators can be described by products of Pauli operators and mutually commute.
Even in the case of the presence of fermionic degrees of freedom, we can express them as the Pauli products by appropriate qubit transformation such as Jordan--Wigner transformation.
The physical Hilbert space is defined as the common +1 eigenspace of these operators, $G_n \ket{\psi} = +\ket{\psi}, \quad \forall n $, where $n$ is the index of vertices in general lattice setup.
Thus, all Gauss's law operators of $\mathbb{Z}_2$ LGTs in general dimensions and any boundary condition can also be represented as a set of commuting Pauli operators.

\section{QMETTS algorithm} \label{sec:QMETTS}
In this section, we review the Quantum Minimally Entangled Typical Thermal States (QMETTS) algorithm.
QMETTS is the quantum extension of the Minimally Entangled Typical Thermal States (METTS) algorithm, which was originally proposed by White \textit{et al.}~\cite{White:2010api, white2009minimally} in the context of tensor network methods.
Motta \textit{et al.}~\cite{Motta:2019yya} adapted this framework to quantum computing and introduced the QMETTS algorithm.
It provides a method for estimating expectation values of observables in finite-temperature and finite-density quantum many-body systems.
In this section, we review the standard formulation of QMETTS without gauge constraints or shot noise.
Gauge constraints will be incorporated into our modified scheme in Sec.~\ref{sec:MUPB}, and shot noise will be discussed separately in Sec.~\ref{sec:shot_analysis}.

To explain the QMETTS algorithm, we rewrite the expectation value $\braket{\Obs}_{\beta, \mu}= \Tr[\Obs e^{-\beta (H-\mu N)} / Z]$ in terms of an orthonormal basis $\{\ket n\}$
\begin{align}
    \braket{\Obs}_{\beta, \mu} 
    &= \sum_{\ket{i} \in \{\ket{n}\}} \frac{1}{Z} \bra{i}e^{-\beta (H-\mu N)/2} \Obs e^{-\beta (H-\mu N)/2}\ket{i}  \nonumber\\
    &= \sum_{\ket{i} \in \{\ket{n}\}} \Prob_i \cdot \bra{\phi_i} \Obs\ket{\phi_i}, \label{eq:decomp_thermal_states}
\end{align}
where $\ket{\phi_i} =  e^{-\beta (H-\mu N)/2}\ket{i}/\sqrt{\bra{i}e^{-\beta (H-\mu N)}\ket{i}}$ and $\Prob_i = \bra{i}e^{-\beta (H-\mu N)}\ket{i}/Z$.
Since $\ket{\phi_i}$ is normalized, it represents a pure quantum state, and
$\Prob_i$ forms a probability distribution satisfying $\Prob_i \geq 0$ and $\sum_i \Prob_i = 1$.
When the basis $\{\ket{n}\}$ is chosen to be a product basis, the states $\ket{\phi_i}$ are often expected to have relatively low entanglement.
For this reason, the states $\ket{\phi_i}$ are referred to as Minimally Entangled Typical Thermal States (METTS) following Ref.~\cite{White:2010api}.
Therefore, as described in Eq.~(\ref{eq:decomp_thermal_states}), the thermal expectation value can be interpreted as an average of
the expectation values under each METTS $\bra{\phi_i}\Obs\ket{\phi_i}$ weighted by $\Prob_i$.

The goal of the QMETTS algorithm is to efficiently sample these contributions by generating a Markov chain of basis states $\ket{i}$ whose stationary distribution is $\Prob_i$.
The procedure is as follows.
First, we sample one state $\ket{i}$ from the measurement basis $\{\ket{n}\}$ and apply the  imaginary-time evolution (ITE) operator $e^{-\beta (H-\mu N)/2}$ to $\ket{i}$, preparing the normalized imaginary-time-evolved state $\ket{i} \mapsto \ket{\phi_i} =  e^{-\beta (H-\mu N)/2}\ket{i}/\sqrt{\bra{i}e^{-\beta (H-\mu N)}\ket{i}}$ on a quantum computer.
Next, we perform a projective measurement for the METTS $\ket{\phi_i}$ with the measurement basis $\{\ket{n}\}$ and obtain a new collapse state $\ket{i'}$.
By repeating the processes of the ITE and the projective measurement, one can generate the Markov chain with the stationary distribution $\Prob_i$ because the transition probability from $\ket{i}$ to $\ket{i'}$, $T_{i \rightarrow i'}:= |\braket{i'|\phi_i}|^2$, 
satisfies the detailed balance condition:
\begin{align}
    \frac{T_{i \rightarrow i'}}{T_{i' \rightarrow i}} = \frac{{\bra{i'}e^{-\beta (H - \mu N)}\ket{i'}}}{{\bra{i}e^{-\beta (H - \mu N)}\ket{i}}}=\frac{\Prob_{i'}}{\Prob_{i}}. \nonumber
\end{align}

We generate a Markov chain of length $N_{{\chain}}$, which produces a sequence of sampled METTS $\{\ket{\phi_{i_k}}\}_{k=1}^{N_{\mathrm{chain}}}$.
For the METTS $\ket{\phi_{i_k}}$ generated at the $k$-th step, we compute the observable estimate $\mu_{k} = \bra{\phi_{i_k}}\Obs\ket{\phi_{i_k}}$ (see the left of Fig.~\ref{fig:sampling_efficient}).
Defining the Monte Carlo average $
\bar{\mu} =
({1}/{N_{\mathrm{chain}}})
\sum_{k=1}^{N_{\mathrm{chain}}} \mu_{k}
$,
this estimator satisfies
\begin{align}
    \Exp[\bar \mu] 
    &= 
    \frac{1}{N_{\text{chain}}}\sum_{k=1}^{N_\chain} \Exp [\mu_{k}]
    = \braket{\Obs}_{\beta, \mu}. \nonumber
\end{align}
Thus, $\bar{\mu}$ is an unbiased estimator of $\braket{\Obs}_{\beta,\mu}$.
\begin{figure*}
    \centering
    \includegraphics[width=1.0\linewidth]{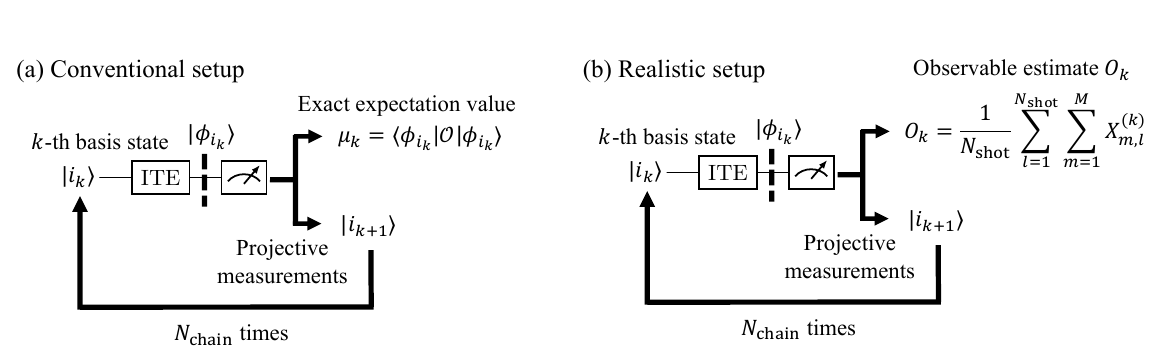}
    \caption{Schematic illustration of two QMETTS measurement schemes.
    Many previous works assume an idealized setting in which shot noise is neglected and expectation values are evaluated exactly (a).
    By contrast, on quantum hardware, expectation values must be estimated from a finite number of measurements.
    We therefore consider a realistic setting (b), where the observable 
    $\Obs=\sum_{m=1}^{M}\Obs_m$ is decomposed into $M$ components, each of which can be estimated using a single measurement circuit.
    For each of the $N_{\chain}$ Markov-chain samples, we perform $N_{\shot}$ measurements for each component, resulting in $N_{\shot}M$ circuit executions for observable estimation, and one additional projective measurement to generate the next basis state.
    The corresponding total number of circuit executions is $N_{\tot}=(N_{\shot}M+1)N_{\chain}$.
    In the case of alternating measurements, $\ket{i_k}$ and $\ket{i_{k+1}}$ are replaced by $\ket{i_k^{(r)}}$ and $\ket{i_{k+1}^{(s)}} $, respectively, where $(r, s)=(1, 2)$ or $(2, 1)$.}
    \label{fig:sampling_efficient}
\end{figure*}

To quantify the statistical error of the Markov chain, one must account for the autocorrelations between successive samples. 
Let $\tau_\mu$ denote the integrated autocorrelation time for the sequence of $\{\mu_{k}\}_{k=1}^{N_\chain}$.
For a sufficiently long Markov chain, the variance of the Monte Carlo average $\bar{\mu}$ is then given by~\cite{madras1988pivot, sokal1997monte, wolff2004monte},
\begin{equation}
\text{Var}[\bar{\mu}] = \frac{\tau_\mu}{N_{\text{chain}}} \sigma_{\mu}^2, \label{eq:metts_variance}
\end{equation}
where $\sigma_\mu^2 = \sum_i \Prob_i \cdot \bra{\phi_i}\Obs\ket{\phi_i}^2 - (\sum_i \Prob_i \cdot \bra{\phi_i}\Obs\ket{\phi_i})^2$ represents the variance of $\bra{\phi_i}\Obs\ket{\phi_i}$ with respect to the METTS distribution $\Prob_i$.
Equation~(\ref{eq:metts_variance}) shows that the Markov chain effectively yields $N_{\text{eff}} = N_{\text{chain}} / \tau_\mu$ independent samples. The integrated autocorrelation time $\tau_\mu$ is defined by
\begin{equation}
\tau_\mu = 1 + 2\sum_{t=1}^{\infty} \rho_\mu(t), \quad \rho_\mu(t) = \frac{\text{Cov}(\mu_{k}, \mu_{k+t})}{\sigma_{\mu}^2}, \label{eq:auto-corr}
\end{equation}
where $\rho_\mu(t)$ is the normalized autocorrelation function. 
This function quantifies the degree of correlation between samples. 
For completely independent samples, $\rho_\mu(t) = 0$ for all $t\ge1$, resulting in $\tau_\mu = 1$.

In practice, however, using only a single measurement basis cannot be efficient for two distinct reasons.
First, if the chosen measurement basis consists of states with definite values of a conserved charge $Q$ satisfying $[Q,H]=0$, transitions between different charge sectors are forbidden. 
For example, if $Q\ket{i}=q_i\ket{i}$, $Q\ket{j}=q_j\ket{j}$, and $q_i\neq q_j$, then $\ket{\phi_i} \propto e^{-\beta H/2}\ket{i}$ remains in the $q_i$ sector, so $|\braket{j|\phi_i}|^2=0$. 
Hence the chain can remain trapped in a fixed sector, and ergodicity is not guaranteed.
Second, near the high-temperature limit, where the imaginary-time evolution approaches the identity $e^{-\beta H} \approx I$, using a single measurement basis can repeatedly yield the same collapse state, producing strongly correlated consecutive samples and hence long autocorrelation times, see Refs.~\cite{White:2010api, Binder:2017bxd, bruognolo2015symmetric} for details.

To make the algorithm efficient, one needs to use two measurement bases and alternate between them.
Efficient mixing is promoted when the two bases $\{\ket{n^{(1)}}\}$ and $\{\ket{m^{(2)}}\}$ satisfy
\begin{align}
    |\braket{i^{(1)}|j^{(2)}}|^2 = \frac{1}{d},
    \label{eq:MUB}
\end{align}
for all $\ket{i^{(1)}} \in \{\ket{n^{(1)}}\}$ and $\ket{j^{(2)}} \in \{\ket{m^{(2)}}\}$,
where $d$ denotes the Hilbert-space dimension~\cite{Binder:2017bxd, bruognolo2015symmetric}.
Such bases address both issues by allowing transitions between conserved sectors and suppressing this high-temperature autocorrelation effect.
The most representative bases that satisfy this relation are $X$ and $Z$~(i.e., computational) bases.
In the language of quantum information, two orthonormal bases satisfying Eq.~(\ref{eq:MUB}) are said to be mutually unbiased and are referred to as mutually unbiased bases (MUBs)~\cite{Wootters:1989lkz}.
In summary, the efficient algorithm is as follows:
\begin{enumerate}
    \item Choose one state $\ket{i^{(1)}}$ from a measurement basis $\{\ket{n^{(1)}}\}$.
    \item Generate the METTS by the ITE, $\ket{\phi_i^{(1)}}=  e^{-\beta (H-\mu N)/2}\ket{i^{(1)}}/\sqrt{\bra{i^{(1)}}e^{-\beta (H-\mu N)}\ket{i^{(1)}}}$.
    \item Compute the observable estimate of $\Obs$ under the METTS $\ket{\phi_i^{(1)}}$; $\bra{\phi_i^{(1)}}\Obs\ket{\phi_i^{(1)}}$.
    \item Measure the $\ket{\phi_i^{(1)}}$ with another measurement basis $\{\ket{m^{(2)}}\}$ 
    satisfying the condition of the MUB with $\{\ket{n^{(1)}}\}$, leading to a collapse into a basis state $\ket{j^{(2)}} \in \{\ket{m^{(2)}}\}$ with probability $|\braket{j^{(2)}|\phi_i^{(1)}}|^2$.
    \item Repeat the same step from $2$ to $4$ switching the role of $\{\ket{n^{(1)}}\}$ and $\{\ket{m^{(2)}}\}$ for $N_\chain$ times.
    \item Average the samples $\bra{\phi_i^{(a)}} \Obs\ket{\phi_i^{(a)}} ~(a=1, 2)$. 
\end{enumerate}
If all these steps are simulated classically, one recovers the original METTS algorithm. 
However, we formulate the construction in the quantum-circuit setting, because the gauge-preserving measurement bases introduced below are naturally implemented as quantum circuits and need not be simple product bases.

In the subsequent section, we extend this approach to gauge theories.
The above strategy works in generic quantum many-body systems~\cite{White:2010api}.
In gauge theories, however, we need to preserve gauge invariance.
Starting from an initial state in physical subspace,
the imaginary-time evolved states preserve gauge invariance due to $[H, G_n]=0$.
Nevertheless, generic projective measurements do not respect the gauge constraints
and may collapse a physical state into an unphysical one.
In the next section, we propose novel measurement bases preserving gauge symmetries while satisfying Eq. (\ref{eq:MUB}) within the physical sector.

Moreover, the above review and most previous works do not take into account shot noise in the estimation of expectation values, although such noise is unavoidable in quantum computation.
In Sec.~\ref{sec:shot_analysis}, we extend the above discussion to incorporate shot noise and optimize the total number of circuit runs.

\section{Efficient construction of MUPB in $\mathbb{Z}_2$ lattice gauge theories} \label{sec:MUPB}
In this section, we develop a QMETTS-based method for estimating the expectation value $\braket{\Obs}_{\beta,\mu}$ in $\mathbb{Z}_2$ LGTs.
The central task of this section is to construct measurement bases that 
(1) preserve gauge invariance and (2) remain mutually unbiased within the physical subspace, while still being efficiently implementable.
First, we refer to such a pair of bases as the Mutually Unbiased Physical Bases (MUPB), defined as follows.
\begin{definition}[Mutually Unbiased Physical Bases] \label{def:MUPB}
    Two measurement bases $\{\ket{n^{(1)}}\}$ and $\{\ket{m^{(2)}}\}$ are called Mutually Unbiased Physical Bases if they satisfy the following properties:
    \begin{enumerate}
        \item Every state in the measurement bases $\{\ket{n^{(1)}}\}$ and $\{\ket{m^{(2)}}\}$ is an eigenstate of all Gauss's law operators:
        \begin{align}
            G_n \ket{i^{(a)}} = g_n \ket{i^{(a)}}, ~g_n \in \{\pm 1\} ,~a=1, 2. \nonumber
        \end{align}
        \item Any two physical states $\ket{i^{(1)}}, \ket{j^{(2)}} \in \Hilbert_{\phys}$ satisfy the mutually unbiased condition within physical subspace $\Hilbert_{\phys}$:
        \begin{align}
            |\braket{i^{(1)}|j^{(2)}}|^2 = \frac{1}{d_{\phys}}, \nonumber
        \end{align}
        where $d_{\text{phys}}$ denotes the dimension of the physical subspace $\mathcal{H}_{\text{phys}}$.
    \end{enumerate}
\end{definition}
The first condition ensures compatibility with the gauge constraints, while the second promotes mixing within the physical subspace.
Because each basis state has definite Gauss's law eigenvalues, projective measurements in MUPB do not mix different gauge sectors.
Therefore, if the initial state is physical, the resulting Markov chain remains entirely within the physical sector and never collapses into a gauge-violating state.

Since measurements in quantum computation are typically performed in the computational basis, MUPB must be implemented by designing appropriate quantum circuits.
However, the definition of MUPB may appear complicated, and it is not obvious how to construct such bases efficiently using quantum circuits.
To address this problem, we use the stabilizer formalism, a mathematical framework originally developed in the context of quantum error correction~\cite{nielsen2010quantum, Gottesman:1997zz}.
In this formalism, logical states are characterized as simultaneous $+1$ eigenstates of a set of commuting Pauli-product operators called stabilizers $\mathcal{S} = \{S_j\}$,
\begin{align}
    S_j \ket{\psi} = + \ket{\psi} ,~ [S_j, S_k] = 0, \quad \text{for} ~\forall S_j, S_k \in \mathcal{S}. \nonumber
\end{align}
In a similar way, as mentioned in Sec.~\ref{sec:LGTsetup}, for $\mathbb{Z}_2$ LGTs in arbitrary spatial dimensions and with arbitrary boundary conditions, the gauge constraints can be viewed as commuting Pauli-product operators that define the physical subspace.
Thus, physical states are precisely the states stabilized by these gauge constraint operators.

Based on this perspective, we use the canonical form of stabilizer groups discussed in Ref.~\cite{Aaronson:2004xuh}. 
This result implies that any stabilizer group can be transformed by a Clifford unitary into a canonical form generated solely by Pauli-$Z$ operators.
Applying this theorem to the gauge constraints, for Gauss's law operators $G_n ~(n=1, \dots, S)$ in $\mathbb{Z}_2$ LGTs in arbitrary spatial dimensions and with arbitrary boundary conditions, we obtain
\begin{align}
    W G_n W^\dagger = Z_n, \qquad  n=1, \dots, S, \nonumber
\end{align}
where $W$ is a classically computable Clifford unitary, see App.~\ref{app:any_MUPB}.
Furthermore, $W$ can be implemented using $\mathcal{O}(N_q^2)$ Clifford gates with depth $\mathcal{O}(\log N_q)$, where $N_q$ denotes the number of qubits.
This transformation allows us to construct the MUPBs shown in Fig.~\ref{fig:gen_MUPB}, and to prove that they satisfy the defining conditions of the MUPB.
Thus, the Clifford circuit required to realize MUPB has gate complexity $\mathcal{O}(N_q^2)$ and depth $\mathcal{O}(\log N_q)$ implying that the construction is efficient.
A proof that the bases shown in Fig.~\ref{fig:gen_MUPB} satisfy the MUPB conditions is given in App.~\ref{app:any_MUPB}.

Regarding the circuit complexity discussed here, we assume all-to-all connectivity 
and parallel application of Clifford gates. 
Whether the general MUPB construction described above can be implemented efficiently 
under nearest-neighbor connectivity remains unclear. 
Nevertheless, for the specific geometry introduced in Sec.~\ref{sec:LGTsetup}, 
we identify a specialized realization that uses only nearest-neighbor gates 
and has constant depth. 
This geometry-adapted realization is obtained heuristically using the same 
stabilizer-based perspective and satisfies the same MUPB conditions. 
It is described and validated in Sec.~\ref{sec:numerical}, and is used in the numerical 
simulations below.
\begin{figure*}
    \centering
    \includegraphics[width=0.9\linewidth]{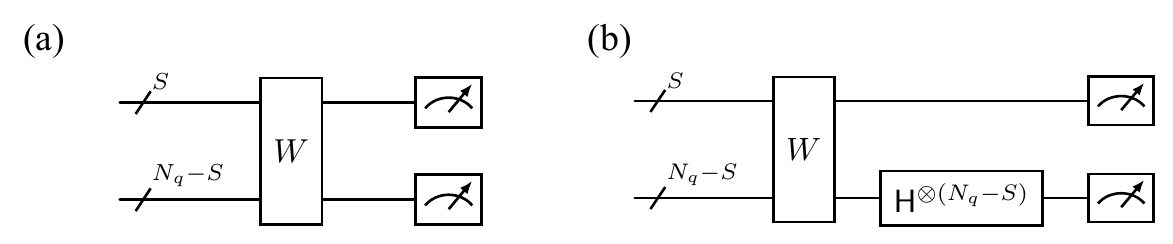}
    \caption{Quantum circuits for implementing the MUPB.
    $W$ is a classically computable Clifford unitary that can be implemented using $\mathcal{O}(N_q^2)$ Clifford gates with depth $\mathcal{O}(\log N_q)$.
    Circuit (a) implements the first measurement basis by applying $W$ before computational-basis measurements.
    Circuit (b) implements the second measurement basis by applying $W$, followed by Hadamard gates on the remaining $N_q - S$ qubits, before computational-basis measurements.}
    \label{fig:gen_MUPB}
\end{figure*}

The above stabilizer-based viewpoint also clarifies why the present construction is particularly natural in QMETTS compared with METTS.
Since Clifford circuits contain entangling gates such as CNOT, the corresponding measurement bases need not be product bases.
In the QMETTS setting, such gauge-preserving bases can nevertheless be implemented directly as measurement circuits.

\subsection{Contribution and relation to prior work} \label{subsec:contribution}
The calculation of thermal equilibrium properties while preserving symmetries
has been extensively studied within the framework of the
METTS algorithm,
particularly for systems with one or a few global conserved quantities by exploiting symmetry-adapted measurement bases in the context of condensed matter.

In Ref.~\cite{Binder:2017bxd}, efficient mixing bases are achieved within the METTS framework
by employing Fourier-transformed bases and Haar-random bases
that respect global symmetries.
While these constructions lead to good ergodic properties and short autocorrelations,
they typically require high complexity.

Another related approach was proposed in Ref.~\cite{goto2019quasiexact},
where nontrivial measurement bases are generated via time evolution
under commuting operators with generators corresponding to symmetries within the METTS framework.
This method does not generally guarantee efficient mixing of the Markov chain.
Moreover, it remains unclear how this construction can be extended to accommodate local gauge constraints, which impose much more stringent conditions than global symmetries. 

In contrast, our MUPB approach realizes mutually unbiased measurements within the physical subspace while preserving gauge invariance, and the corresponding basis changes can be implemented using shallow Clifford circuits.
This provides a controlled way to enhance mixing within the physical subspace without sacrificing the redundant Hilbert-space structure.

\section{Shot-noise analysis} \label{sec:shot_analysis}
Unlike tensor-network methods, quantum computation inherently involves shot noise when estimating expectation values of observables.
In this sense, the discussion in Sec.~\ref{sec:QMETTS} corresponds to an idealized setting.
Nevertheless, many previous analyses of the QMETTS algorithm assume exact expectation-value estimation, in close analogy with the ideal METTS framework.
In this section, we extend the QMETTS analysis to include shot noise and investigate efficient strategies for observable estimation in realistic quantum-computing settings.
Here, we address the following two questions: in the presence of shot noise, does the QMETTS estimator reproduce finite-temperature and finite-density expectation values, and what is the optimal number of observable-measurement shots per sampled collapse state when the total number of circuit executions used to measure observables and to generate the next collapse states in the Markov chain is fixed?
Note that the analysis presented in this section is not limited to gauge systems and can be easily extended to the cases with alternating measurements.
Throughout this analysis, we assume that the Markov chain is initialized from the stationary distribution.
We also neglect gate noise, state-preparation errors, and ITE errors.

To address the questions, we consider the following setup.
An observable $\Obs$ is decomposed as
$
    \Obs = \sum_{m=1}^M \Obs_m ,
$
where each $\Obs_m$ denotes a group of Pauli terms that are measured in the same measurement setting.
Here, possible coefficients of Pauli strings are absorbed into the definition of $\Obs_m$.
For a given METTS sample $\ket{\phi_{i_k}}$, we estimate each $\Obs_m$ using $N_{\shot}$ independent measurements on independent copies of the same state.
We denote by $X_{m,l}^{(k)}$ the $l$-th single-shot measurement outcome for $\Obs_m$ under the $k$-th METTS sample, where $l=1,\dots,N_{\shot}$.
See the right illustration in Fig.~\ref{fig:sampling_efficient}.
Thus, conditioned on the METTS sample $\ket{\phi_{i_k}}$, the measurement outcomes for distinct pairs $(m,l)$ are statistically independent.
The estimator for the observable under the $k$-th METTS sample is then given by
\begin{align}
    O_k
    =
    \frac{1}{N_{\shot}}
    \sum_{l=1}^{N_{\shot}}
    \sum_{m=1}^M
    X_{m,l}^{(k)} . \nonumber
\end{align}
Repeating this procedure along the METTS Markov chain yields a sequence of noisy observations $\{O_k\}_{k=1}^{N_{\chain}}$.
The corresponding Markov-chain estimator is defined as
\begin{align}
    \bar O
    =
    \frac{1}{N_{\chain}}
    \sum_{k=1}^{N_{\chain}} O_k . \nonumber
\end{align}
Thus, the number of circuit executions used for observable estimation is
$
    N_{\est} = N_{\shot}  N_{\chain} M.
$
Since one additional projective measurement is required to generate each subsequent METTS sample, the total number of circuit executions is
\begin{align}
    N_{\tot}
    =
    N_{\est}+N_{\chain}
    =
    (N_{\shot}M+1)N_{\chain}.
    \label{eq:N_tot}
\end{align}

Now, we are in a position to answer the first question: whether the QMETTS estimator is unbiased.
The sequence $\{O_k\}_{k}$ contains two sources of statistical fluctuation: the METTS-sampling fluctuations and the finite-shot measurement noise, making the analysis nontrivial.
To separate these two sources, we first average over the shot noise while keeping the METTS samples fixed.
Conditioned on the $k$-th METTS sample $\ket{\phi_{i_k}}$, the expectation value of the single-shot outcome $X_{m,l}^{(k)}$ is given by
\begin{align}
    \Exp_{\shot}\!\left[X_{m,l}^{(k)} \mid \phi_{i_k}\right]
    =
    \bra{\phi_{i_k}}\Obs_m\ket{\phi_{i_k}} \nonumber
\end{align}
where $\Exp_{\shot}[\cdot | \phi_{i_k}]$ denotes the expectation over the measurement outcomes conditioned on the fixed METTS sample $\ket{\phi_{i_k}}$.
Therefore,
\begin{align}
    \Exp_{\shot}\!\left[O_k \mid \phi_{i_k}\right]
    &=
    \sum_{m=1}^M \bra{\phi_{i_k}}\Obs_m\ket{\phi_{i_k}}
    =
    \bra{\phi_{i_k}}\Obs\ket{\phi_{i_k}}. \nonumber
\end{align}
Using the law of total expectation, we then obtain
\begin{align}
    \Exp[\bar O]
    &=
    \Exp_{\chain}
    \left[
        \Exp_{\shot}
        \left[
            \bar O
            \mid
            \{\phi_{i_k}\}_{k=1}^{N_{\chain}}
        \right]
    \right] \nonumber \\
    &=
    \sum_i \Prob_i \bra{\phi_i}\Obs\ket{\phi_i}
    =
    \braket{\Obs}_{\beta, \mu}, \nonumber
\end{align}
where $\Exp_{\chain}[\cdot]$ denotes the expectation over the METTS Markov chain.
Thus, under the stationary distribution, the estimator remains unbiased for any $N_{\shot}$ and $N_{\chain}$.

We next evaluate the variance of the Markov-chain estimator $\bar O$.
As in the expectation-value calculation, we first consider the shot-noise variance for fixed METTS samples.
For a fixed METTS sample $\ket{\phi_{i_k}}$, we define the conditional shot-noise variance $\sigma_{\shot|\phi_{i_k}}^2$ by
\begin{align}
    \sigma_{\shot|\phi_{i_k}}^2 := \sum_{m=1}^M \left[\bra{\phi_{i_k}} \Obs_m^2 \ket{\phi_{i_k}} - (\bra{\phi_{i_k}} \Obs_m \ket{\phi_{i_k}})^2 \right]. \nonumber
\end{align}
Then, the shot noise variance of $O_k$ conditioned on METTS $\ket{\phi_{i_k}}$ is given by
\begin{align}
    \Var_\shot[O_k |\phi_{i_k}] = \frac{1}{N_\shot^2} \sum_{l=1}^{N_\shot} \sigma_{\shot|\phi_{i_k}}^2 = \frac{\sigma_{\shot|\phi_{i_k}}^2}{N_\shot}. \nonumber 
\end{align}
Since the conditional single-shot variance $\sigma_{\shot|\phi_{i_k}}^2$ depends on the METTS sample, we define its average over the METTS distribution as $\sigma_\shot^2:= \sum_i \Prob_i \sigma_{\shot|\phi_i}^2$.
Using the law of the total variance, the variance of $\bar O$ with $N_\shot$ measurements is given by,
\begin{align}
    \Var[\bar O](N_\shot) 
    &= \Var_\chain\left[\Exp_\shot\left[
            \bar O
            \mid
            \{\phi_{i_k}\}_{k=1}^{N_{\chain}}
        \right]\right] \nonumber \\
    &\quad + \Exp_\chain\left[\Var_\shot\left[
            \bar O
            \mid
            \{\phi_{i_k}\}_{k=1}^{N_{\chain}}
        \right]\right] \nonumber \\
    &= \frac{\tau(N_\shot)}{N_\chain} \left( \sigma_\mu^2 + \frac{\sigma_\shot^2}{N_\shot} \right), \label{eq:var_w_shot}
\end{align}
where $\tau(N_\shot)$ is the integrated autocorrelation time of the observed sequence $\{O_k\}_{k=1}^{N_\chain}$,
\begin{align}
    \tau(N_\shot) = 1 + 2\sum_{t=1}^{\infty} \rho_{O}(t, N_\shot), \nonumber  \\
    \rho_O(t, N_\shot) = \frac{\text{Cov}(O_k, O_{k+t})}{\sigma_{\mu}^2 + (1/N_\shot) \sigma_\shot^2}. \label{eq:auto_corr_w_shot}
\end{align}
When $N_\shot$ is sufficiently large that the shot noise is negligible, $\tau (N_\shot)$ is well-approximated by the autocorrelation time of the original Markov chain $\tau_\mu$ in Eq.~(\ref{eq:auto-corr}).

This variance formula allows us to address the second question: what is the optimal choice of $N_\shot$ at fixed $N_\tot$?
Using Eq.~(\ref{eq:var_w_shot}) together with the budget constraint in Eq.~(\ref{eq:N_tot}), we minimize the estimator variance with respect to $N_\shot$.
We find
\begin{align}
    N_{\shot}^* \sim \max\!\left(1,\sqrt{\frac{\sigma_\shot^2}{M\sigma_\mu^2 \tau_\mu}}\right), \label{eq:opt_shot}
\end{align}
as derived in App.~\ref{app:autocorr}.
This expression shows that fewer shots are favored when the autocorrelation is strong, when the METTS-sampling variance $\sigma_\mu^2$ is large, or when the number of measurement groups $M$ is large.
Conversely, when the uncertainty originating from shot noise $\sigma_\shot^2$ is large, it is advantageous to increase $N_{\shot}$.

In practice, however, the quantities $\sigma_\mu^2$, $\sigma_\shot^2$, and $\tau_\mu$ are generally not known in advance.
Nevertheless, we find that the variance ratio between the single-shot strategy and the optimal-shot strategy satisfies
\begin{align}
    1\le \frac{\Var[\bar O](1)}{\Var[\bar O](N_\shot^*)} < 1+ \frac{1}{M} \le 2, \label{eq:single_shot_ineq}
\end{align}
for any $N_\shot^*$ (see App.~\ref{app:autocorr}).
Thus, even without prior knowledge of the optimal parameters, the single-shot strategy provides a robust choice, with a variance at most twice the optimum.

It is also useful to compare the single-shot strategy with direct independent measurements on the Gibbs state.
Although the general expression is more involved, it takes a simple form for $M=1$ (see App.~\ref{app:autocorr} for $M\neq 1$ cases).
In this case, the total number of circuit executions is $N_{\tot}=2N_{\chain}$, and the variance of the single-shot QMETTS estimator can be written as
\begin{align}
    \Var[\bar O](1)
    =
    2\tau(1)\Var_{\Gibbs}\left[\frac{1}{N_\tot} \sum_{l=1}^{N_\tot} Y_{l}\right],
\end{align}
where $Y_l$ denotes the $l$-th independent shot measurement outcome of $\Obs$ on the Gibbs state.
Here, $\Var_{\Gibbs}[\cdot]$ denotes the variance with respect to independent measurements performed on the Gibbs state, so that
$\Var_{\Gibbs}\left[(1/N_\tot) \sum_{l=1}^{N_\tot}Y_l\right]$
is the variance of the estimator obtained from $N_{\tot}$ independent Gibbs-state measurements.
Thus, for $M=1$, the single-shot QMETTS estimator differs from the corresponding Gibbs-state estimator by a factor of $2\tau(1)$.
In particular, if the observable sequence exhibits sufficiently strong negative correlations such that $\tau(1) < 0.5$, the QMETTS variance can be smaller than that of direct Gibbs-state measurements.
However, since the integrated autocorrelation time is generally not known before performing numerical simulations, it is difficult to design a strategy that is guaranteed in advance to outperform direct Gibbs-state measurements.
We emphasize that this should not be regarded as a direct comparison of resource costs: it assumes access to independent measurements on the Gibbs state, even though the Gibbs state is mixed and cannot be prepared from a pure initial state by a unitary acting on the system alone, whereas QMETTS avoids this requirement by sampling an ensemble of pure states.

\begin{figure*}
    \centering
    \includegraphics[width=0.9\linewidth]{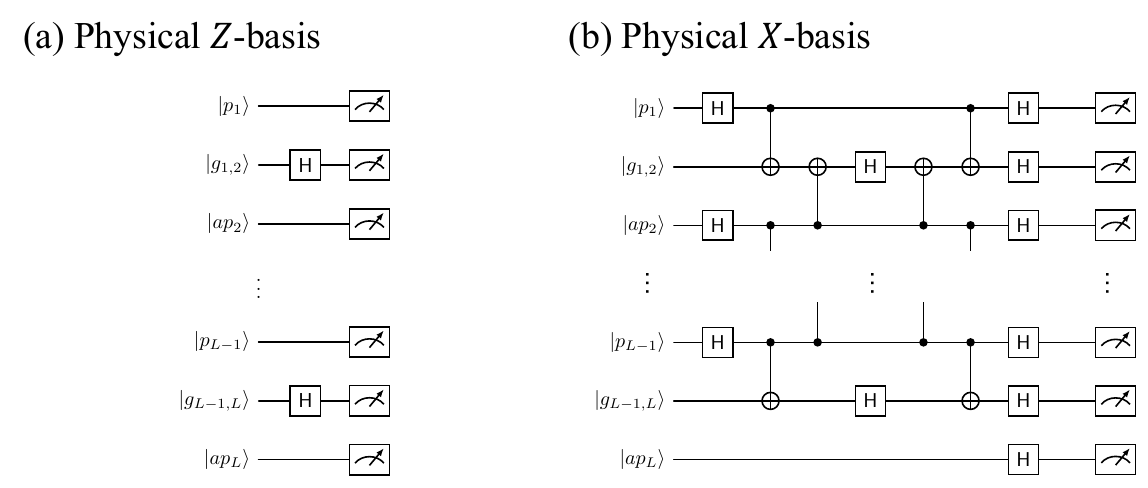}
    \caption{Quantum circuits for implementing the MUPB in the $\mathbb{Z}_2$ LGT setup introduced in Sec.~\ref{sec:LGTsetup}.
    The left and right circuits correspond to measurements in the physical $Z$ basis and the physical $X$ basis, respectively.
    Each circuit shows the measurement setup for the local degrees of freedom, including the gauge-link field $\ket{g_{n-1,n}}$ and the matter fields $\ket{p_n}$ and $\ket{ap_n}$.
    The right circuit implements the transformation to the complementary physical basis using CNOT and Hadamard gates.
    }
    \label{fig:MUPB}
\end{figure*}

\section{Numerical simulation} \label{sec:numerical}
In this section, we investigate the numerical performance of our algorithm, namely QMETTS with MUPB and the single-shot strategy, by applying it to the $\mathbb{Z}_2$ LGT setup introduced in Sec.~\ref{sec:LGTsetup}.
\rev{We first provide an explicit construction of the MUPB for the model considered here and next describe the numerical setup, including the simulation parameters and measurement protocol.
We then present the results obtained using classically computed exact imaginary-time evolution (exact ITE).
Finally, towards the practical implementation of ITE on quantum hardware, we combine our measurement scheme with the QITE algorithm proposed by Motta \textit{et al.} and heuristically assess its performance.}

\subsection{Concrete MUPB construction in the present setup}
\begin{figure*}
    \centering
    \includegraphics[width=0.95\linewidth]{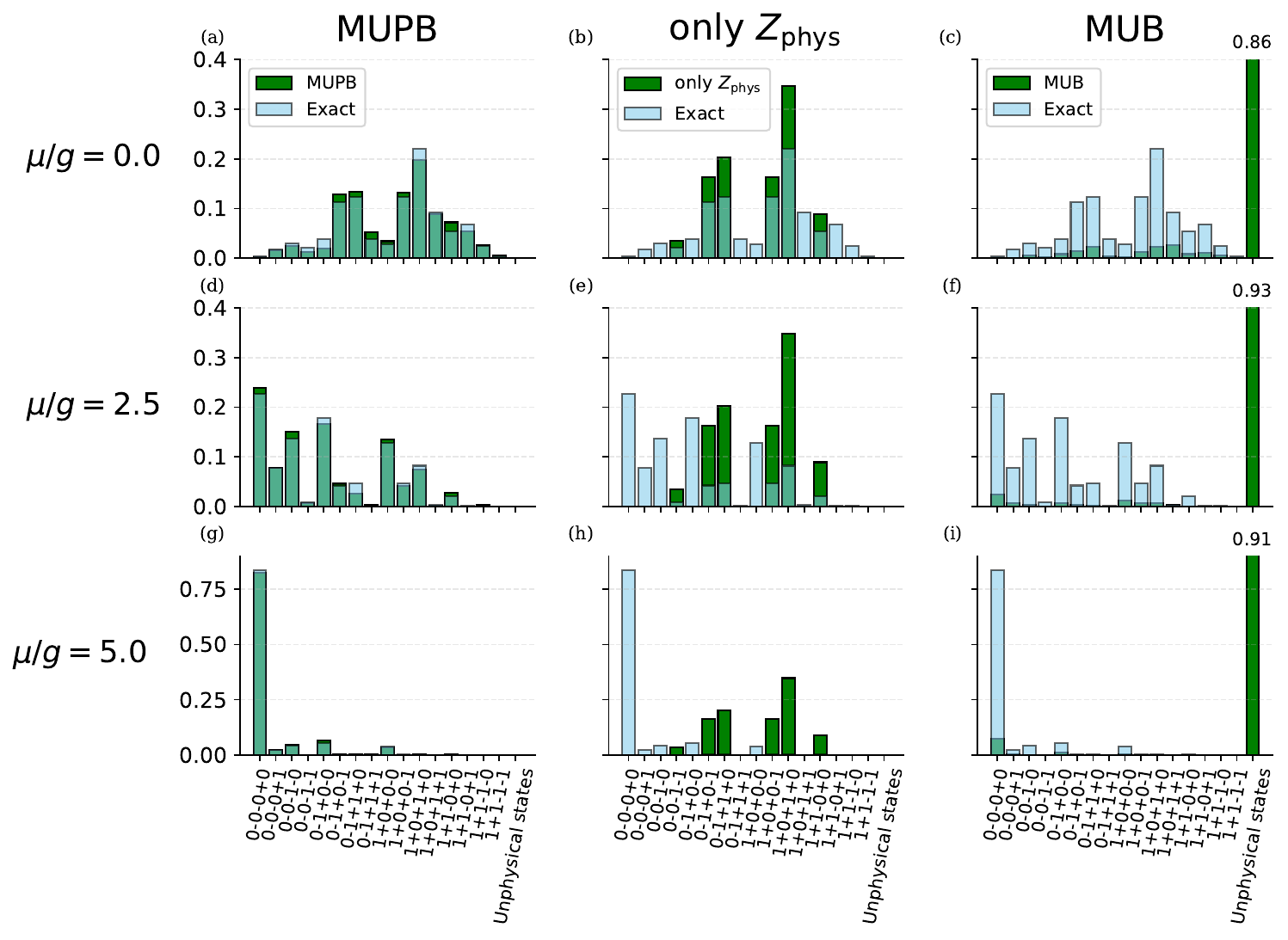}
    \caption{
    Sampling distributions of collapse states obtained from projective measurements during QMETTS sampling (green) and the exact probabilities (light blue) for $L_{\mathrm{KS}}=4$.
    \rev{The ITE is implemented exactly using classical computation.}
    The panels are arranged by chemical potential (rows) and measurement scheme (columns).
    The top panels (a--c), middle panels (d--f), and bottom panels (g--i) correspond to $\mu/g=0.0$, $2.5$, and $5.0$, respectively, while the left, middle, and right columns correspond to MUPB, only $Z_{\mathrm{phys}}$\rev{, and a MUB pair consisting of the physical $Z$-basis and its mutually unbiased counterpart}, respectively.
    The inverse temperature is fixed at $\beta g=1.0$.
    For the MUPB and MUB runs, we plot only the collapse states obtained at the physical $Z$-basis measurement steps of the alternating Markov chain.
    All gauge-violating collapse states are grouped into ``unphysical states''.
    The exact probabilities are calculated as $\Prob_i = \langle i^{(Z_\phys)}| e^{-\beta(H-\mu N)} |i^{(Z_\phys)} \rangle / Z$.
    For all panels, the initial state is fixed to $\ket{1+0+1+0}$, which corresponds to the trivial no-excitation state in the qubit description introduced in Sec.~\ref{sec:LGTsetup}.
    All QMETTS results are obtained with $N_\chain=1000$.}
    \label{fig:hist_L4}
\end{figure*}
We explicitly construct a pair of MUPB for the $\mathbb{Z}_2$ LGT setup in Sec.~\ref{sec:LGTsetup}. 
Here, using the same stabilizer-based framework as in Sec.~\ref{sec:MUPB}, we identify a geometry-adapted realization that satisfies the same MUPB conditions but admits a substantially simpler nearest-neighbor, constant-depth implementation, thereby retaining efficient scaling to larger systems, shown in Fig.~\ref{fig:MUPB}.
The numerical simulations below use this specialized realization.
We refer to these bases as the physical $Z$-basis (left in Fig.~\ref{fig:MUPB}) and the physical $X$-basis (right in Fig.~\ref{fig:MUPB}), and denote their basis states by $\{\ket{e_i^{(Z_{\text{phys}})}}\}$ and $\{\ket{e_i^{(X_{\text{phys}})}}\}$, respectively.
The circuit for the physical $Z$-basis consists of Hadamard gates acting on the gauge links.
The physical $X$-basis is obtained by supplementing these Hadamard gates with a unitary operator $U$, defined by
\begin{align}
    U &:= V^{\dagger} \mathsf{H}_{L_{\text{KS}}} \cdot \left( \prod_{m=1}^{L_{\text{KS}}-1} \sigma_{m,m+1}^\mathsf{H} \right) V, \label{eq:physX_U}\\
    V &:=  \left( \prod_{m=1}^{L_{\text{KS}}-2} \CNOT_{f_{m+1}, g_{m,m+1}} \right) \nonumber \\
    &~~~~~\cdot \left( \prod_{m=1}^{L_{\text{KS}}-2} \CNOT_{f_m, g_{m,m+1}} \right)
    \cdot \left( \prod_{m=1}^{L_{\text{KS}}-1} \mathsf{H}_m \right), \label{eq:physX_V}
\end{align}
where $\CNOT_{f_p,g_{q,q+1}}$ denotes a CNOT gate with control qubit corresponding to $p$-th fermionic site $f_p$ and target qubit corresponding to $(q, q+1)$ gauge link $g_{q,q+1}$.
$\mathsf{H}_p$ and $\sigma^{\mathsf{H}}_{p,p+1}$ denote Hadamard gates acting on site $p$ and link $(p,p+1)$, respectively.
The operator $U$ plays a central role in generating mixing within the physical subspace while preserving the local gauge symmetries.
A detailed proof that the physical $Z$- and $X$-bases satisfy the MUPB conditions is given in App.~\ref{app:proof}.

\subsection{Simulation and measurement setup}
First, we describe the simulation parameters used in our calculations.
In our simulations, all parameters are rescaled by the gauge coupling $g$, with $ag=0.25$ and \rev{$m/g=0.05$}.
All numerical simulations were performed using Qiskit.

Next, we describe the measurement setup used in the following numerical simulations.
We focus on three observables: the energy density
$\braket{h}_{\beta,\mu}:=\braket{H}_{\beta,\mu}/L_D$, the chiral condensate
$\braket{\bar{\psi}\psi}_{\beta,\mu}$, and the \rev{fermion} number density
$\braket{n}_{\beta,\mu}:=\braket{N}_{\beta,\mu}/L_D$.
The chiral condensate is represented in the qubit basis as
\begin{align}
    \bar{\psi}\psi
    := \frac{1}{L_D}\sum_{n=1}^{L_{\mathrm{KS}}} (-1)^{n\rev{+1}}
    \chi_n^\dagger \chi_n
    = \frac{1}{L_{\mathrm{KS}}}\sum_{n=1}^{L_{\mathrm{KS}}}
    (-1)^{n\rev{+1}} Z_n .
    \nonumber
\end{align}
It serves as a useful indicator of the finite-density transition in this
model.  
The \rev{fermion} number density measures the \rev{fermion} number relative to the
staggered vacuum.

To measure the energy density, we decompose the Hamiltonian into three parts,
\begin{align}
H_X &= -\frac{1}{4a} \sum_{n=1}^{L_{\mathrm{KS}}-1}
X_n \sigma^Z_{n,n+1} X_{n+1}, \nonumber \\ 
H_Y &= -\frac{1}{4a} \sum_{n=1}^{L_{\mathrm{KS}}-1}
Y_n \sigma^Z_{n,n+1} Y_{n+1}, \nonumber \\
H_D &= a g^2 \sum_{n=1}^{L_{\mathrm{KS}}-1}
\left( 1 - \sigma^X_{n,n+1} \right)
+ \frac{m}{2} \sum_{n=1}^{L_{\mathrm{KS}}} (-1)^{n\rev{+1}} Z_n ,
\nonumber
\end{align}
so that $H=H_X+H_Y+H_D$.  
The Pauli terms within each of $H_A$ $(A=X,Y,D)$ are qubit-wise commuting and can therefore be measured using a single measurement setting respectively.
With this grouping, we measure the energy using $M=3$ measurement groups.
In contrast, the chiral condensate and the \rev{fermion}
number density contain only $Z$ operators, and we measure each of them using a single measurement setting, i.e., $M=1$.

\subsection{\rev{Numerical results with exact ITE}}
\begin{figure*}[t]
    \centering
    \includegraphics[width=0.95\linewidth]{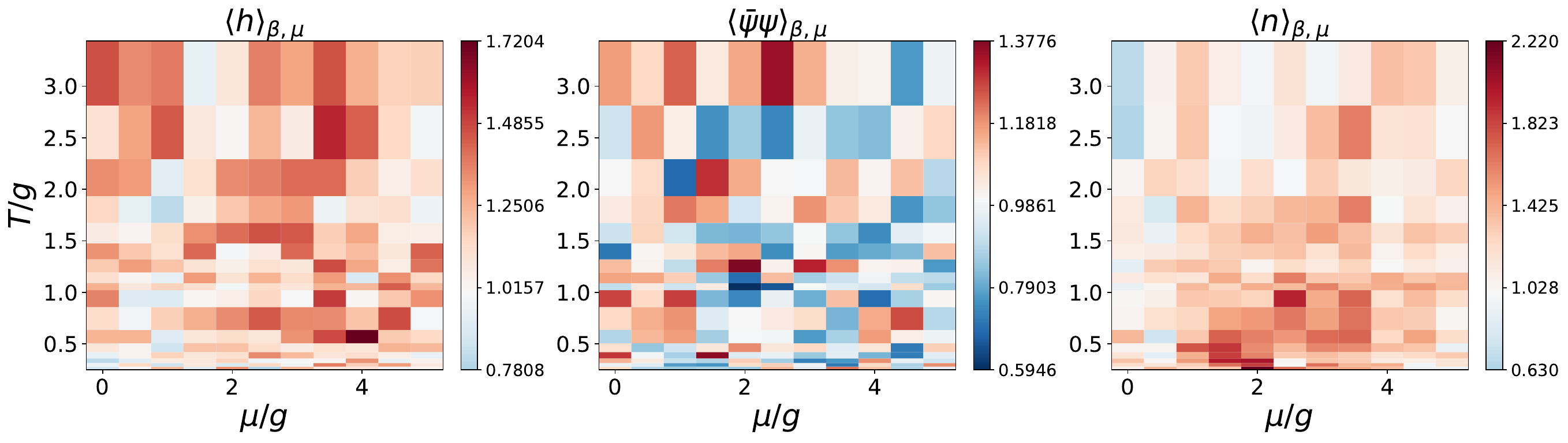}
    \caption{
    Integrated autocorrelation times $\tau(1)$ for the energy density $\braket{h}_{\beta,\mu}$ (left), chiral condensate $\braket{\bar \psi \psi}_{\beta,\mu}$(middle), and \rev{fermion} number density $\braket{n}_{\beta, \mu}$ (right), shown over the $(T/g,\mu/g)$ parameter space.
    The system size is \rev{$L_{\text{KS}}=8$}.
    The autocorrelation sum is truncated using the windowing procedure.
    The window parameter $w=3$.
    $\tau(1)=1$ corresponds to independent sampling.
    }
    \label{fig:auto_correlation}
\end{figure*}
In this section, we present the numerical results obtained using our proposed algorithm.
\rev{To isolate the performance of our sampling scheme from errors in ITE, we first implement the ITE exactly by classical computation, where the imaginary-time evolution operator $e^{-\beta (H - \mu N)/2}$ is applied directly without any approximations.}
We first consider QMETTS with three different measurement schemes: the constructed MUPB, only the physical $Z$-basis\rev{, and a pair of mutually unbiased bases, one of which is the physical $Z$-basis.
The additional basis can be constructed by applying the quantum circuit
$\prod_{m=1}^{L_{\mathrm{KS}}}\mathsf{H}_m$,
which applies Hadamard gates to the fermionic matter qubits, before measurement in the computational basis.
Together with the physical $Z$-basis, this additional basis forms a pair of mutually unbiased bases, i.e. satisfying Eq.~(\ref{eq:MUB}).
Unlike the MUPB, however, the pair is not restricted to the physical subspace.
We refer to this measurement scheme as the MUB scheme below.}

We compare these three measurement schemes by examining the sampling distributions of the collapse states for a system size $L_{\mathrm{KS}}=4$, corresponding to $N_q=7$ qubits.
We fix the Markov-chain length at $N_{\chain}=1000$ and choose the initial state $\ket{1+0+1+0}$ in the qubit representation introduced in Sec.~\ref{sec:LGTsetup}, which corresponds to the trivial configuration without excitations.
Figure~\ref{fig:hist_L4} shows the results obtained using the MUPB scheme (left column), only the physical $Z$-basis (middle column)\rev{, and the MUB scheme (right column)}, for three representative chemical potentials, $\mu/g=0.0$, $2.5$, and $5.0$, with $\beta g=1.0$.
For the MUPB \rev{and MUB} schemes, we show only the collapse states obtained at the physical $Z$-basis measurement steps of the alternating Markov chains, so that they can be directly compared with the exact stationary distribution in the physical $Z$-basis, $\Prob_i = \bra{i^{(Z_{\text{phys}})}} e^{-\beta (H-\mu N)} \ket{i^{(Z_{\text{phys}})}}/Z$.
In Fig.~\ref{fig:hist_L4}, all gauge-violating states are grouped into ``unphysical states.''

\rev{First, the MUB scheme fails to reproduce the exact stationary distribution.
In particular, it predominantly samples gauge-violating states because the states in the additional measurement basis are not eigenstates of the Gauss’s law operators.}
On the other hand, the MUPB scheme and the scheme using only the physical $Z$-basis restrict the sampling to the physical sector because their measurement bases are compatible with gauge invariance.
However, the approach using only the physical $Z$-basis exhibits significant deviations from the exact stationary distribution.
This indicates that using only the physical $Z$-basis does not ensure ergodicity.
This lack of mixing can be understood from the conservation of the number operator $N$.
The initial state $\ket{1+0+1+0}$ belongs to the $N=0$ sector, and the imaginary-time evolution preserves this sector because $[N,e^{-\beta(H-\mu N)}]=0$.
Moreover, measurements in the physical $Z$-basis collapse the state onto
eigenstates of $N$.
Thus, when only the physical $Z$-basis is used, the Markov chain remains trapped in the fixed-$N$ sector in the present setup.
By contrast, the MUPB-based QMETTS accurately reproduces the exact stationary distribution, consistent with the improved mixing expected from the MUPB.

In App.~\ref{app:init_depend}, we further tested several distinct initial states using the MUPBs and found that the sampled collapse-state distributions after $N_\chain=1000$ steps are all consistent with the exact stationary distribution within finite-sampling fluctuations.
This observation indicates that, in the present finite-size benchmark with the MUPB sampling scheme, the effect of burn-in is not practically significant at the chain lengths considered~
\footnote{
In Figs.~\ref{fig:beta_H} and \ref{fig:mu_dependence}, 
no initial samples are discarded as burn-in, yet the QMETTS estimates for 
$L_{\mathrm{KS}}=8$ remain consistent with the exact results within statistical fluctuations.
This provides additional evidence that burn-in effects are not practically significant 
for the chain lengths and parameter ranges considered here.
}.

Next, \rev{we increase the system size to $L_{\mathrm{KS}}=8$, corresponding to $N_q=15$ qubits,} and investigate the autocorrelations in the Markov chains generated using the MUPB scheme.
For the three observables specified above (the energy density, the chiral condensate, and the \rev{fermion} number density) we estimate the integrated autocorrelation time $\tau(1)$.
For each observable, we use $N_{\rm chain}=1000$ METTS samples and one shot per measurement group.
The number of measurement groups depends on the observable: $M=3$ for the Hamiltonian, whereas $M=1$ for the chiral condensate and the \rev{fermion} number density.
Therefore, the corresponding total numbers of circuit executions are not identical among the observables, although the autocorrelation analysis is performed using Markov chains of the same length.
In practice, the finite-chain estimate of the autocorrelation sum in Eq.~\eqref{eq:auto_corr_w_shot} is susceptible to statistical fluctuations, particularly at large time lags.
To suppress these influences, we introduce a window parameter $w$ and truncate the sum following the standard windowing procedure~\cite{madras1988pivot}.
The integrated autocorrelation time is then estimated as
$
    \tau(1)
    =
    1 + 2\sum_{t=1}^{w} \rho_O(t,1).
$
\rev{We choose the window parameter $w=3$ based on the analysis of the autocorrelation functions by investigating the time where statistically uncorrelated fluctuations become dominant at representative parameters.
See App.~\ref{app:auto_func} for details.}

Figure~\ref{fig:auto_correlation} shows the resulting integrated autocorrelation times.
As shown in Fig.~\ref{fig:auto_correlation}, the autocorrelation behavior is observable dependent: the chiral condensate exhibits negative autocorrelations for many parameter choices, whereas the \rev{fermion} number density tends to show positive autocorrelations as the temperature decreases.
Despite these differences, the autocorrelation times remain moderate throughout the parameter range studied here.
These results suggest that MUPB, combined with the single-shot measurement strategy, achieves efficient mixing in the present finite-size benchmark.
\rev{In App.~\ref{app:system_size_depend_autocorr}, we further investigate the system-size dependence of the integrated
autocorrelation time for $L_{\mathrm{KS}}=2,4,6,$ and $8$.
Within this range of system sizes, we observe no clear systematic
system size dependence.}

Based on these observations, we examine the $\beta$ dependence of the energy density $\braket{h}_{\beta,\mu}$ shown in Fig.~\ref{fig:beta_H} for \rev{$L_{\text{KS}}=8$}.
To illustrate the practical advantage of the single-shot strategy under a fixed total number of circuit executions, we compare $N_{\shot}=1$ and $33$ while keeping the total number of circuit executions fixed at $N_{\tot}=4000$.
The corresponding chain lengths are $N_{\chain}=1000$ for $N_{\shot}=1$ and $N_{\chain}=40$ for $N_{\shot}=33$ by Eq.~(\ref{eq:N_tot}) since $M=3$ for the Hamiltonian.
For both shot settings, the numerical results fluctuate around the exact-diagonalization result shown by the black solid line.
While the difference between the two approaches is small in the low-temperature regime, the $N_{\shot}=1$ strategy yields smaller error bars in the high-temperature regime.
At high temperature, more states contribute appreciably, which increases the METTS-sampling variance $\sigma_\mu^2$.
According to the estimate in Eq.~(\ref{eq:opt_shot}),
an increase in $\sigma_\mu^2$ drives the optimal number of shots toward smaller values.
Therefore, the behavior observed in the high-temperature regime of Fig.~\ref{fig:beta_H} is consistent with our analytical prediction that the single-shot strategy becomes favorable in this regime.
\begin{figure}
    \centering
    \includegraphics[width=0.8\linewidth]{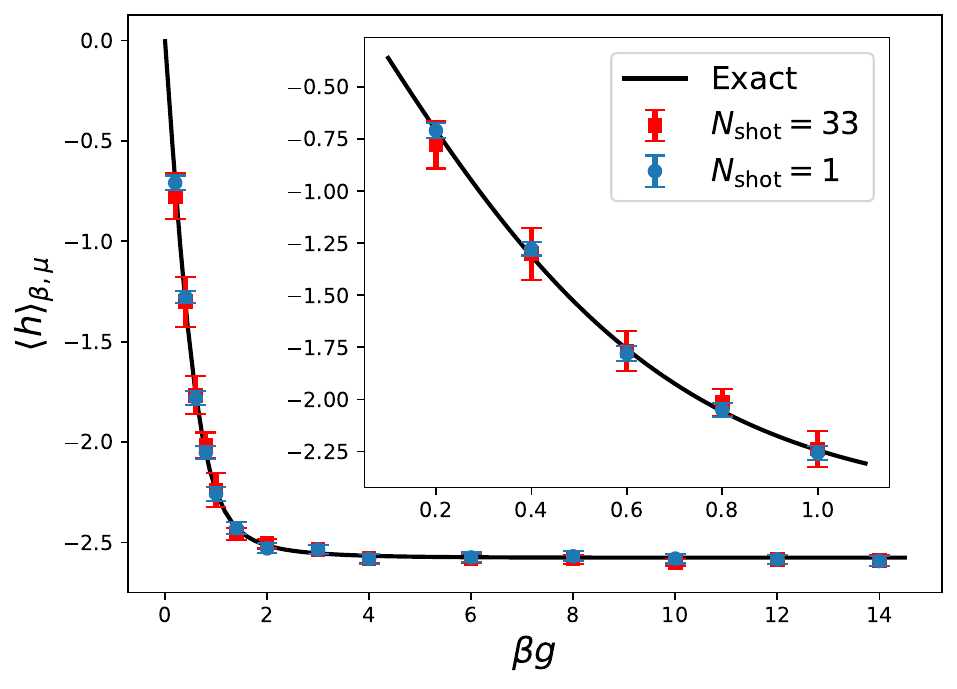}
    \caption{
    Temperature dependence of the energy density $\braket{h}_{\beta,\mu}$ as a function of inverse temperature $\beta g$ for \rev{$L_{\mathrm{KS}}=8$} at zero chemical potential ($\mu/g=0$).
    The QMETTS results for $N_{\shot}=1$ (blue dots with error bars) and $N_{\shot}=33$ (red dots with error bars) are compared with the exact-diagonalization result (black solid line).
    In both settings, the total number of circuit executions is fixed at $N_{\tot}=4000$, so the corresponding chain lengths are $N_{\chain}=1000$ and $40$, respectively.
    The inset highlights the high-temperature region $\beta g \in [0,1.0]$.
    }
    \label{fig:beta_H}
\end{figure}

\begin{figure*}
    \centering
    \includegraphics[width=1.0\linewidth]{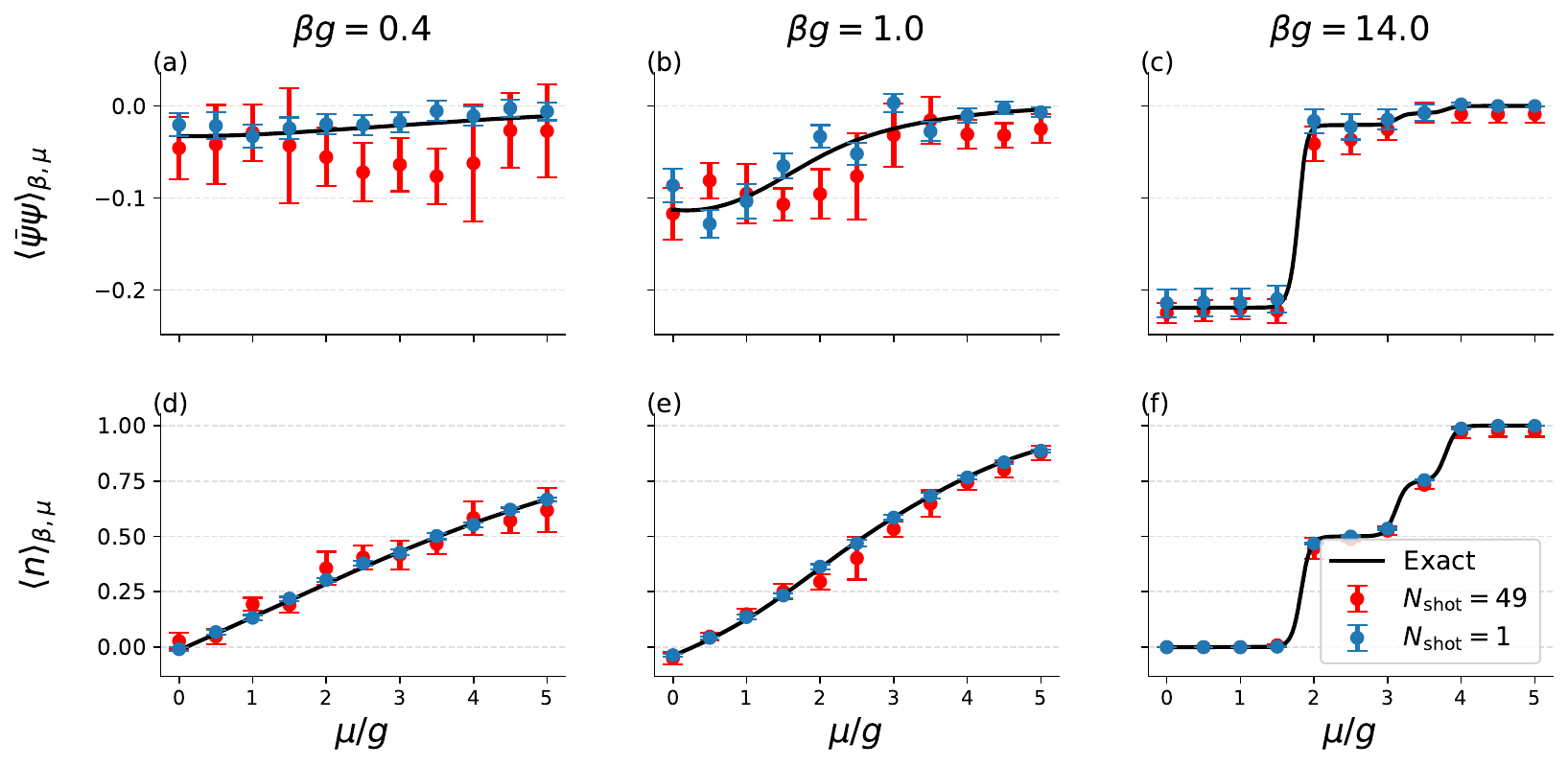}
    \caption{
    Dependence of chiral condensate $\langle \bar{\psi}\psi \rangle_{\beta, \mu}$ and \rev{fermion} number density $\langle n \rangle_{\beta, \mu}$ on the chemical potential $\mu/g$ for \rev{$L_{\text{KS}}= 8$}.
    Top panels (a–c) show the chiral condensate, while bottom panels (d–f) show the \rev{fermion} number density. 
    From left to right, the columns correspond to inverse temperatures $\beta g = 0.4$, $1.0$, and $14.0$, respectively.
    Both observables are measured by a single measurement group setting, $M=1$.
    The QMETTS results for $N_{\shot}=1$ (blue dots with error bars) and $N_{\shot}=49$ (red dots with error bars) are compared with the exact-diagonalization result (black solid line).
    In both settings, the total number of circuit executions is fixed at $N_{\tot}=2000$, so the corresponding chain lengths are $N_{\chain}=1000$ and $40$, respectively.
    }
    \label{fig:mu_dependence}
\end{figure*}
Figure~\ref{fig:mu_dependence} shows the $\mu$ dependence of the chiral condensate $\braket{\bar{\psi}\psi}_{\beta,\mu}$ and the \rev{fermion} number density $\braket{n}_{\beta,\mu}$ at three different temperatures for \rev{$L_{\text{KS}}= 8$} \footnote{
The integrated autocorrelation time can become slightly negative due to statistical fluctuations for the short-chain case with $N_{\chain}=40$.
To avoid underestimating the statistical uncertainty, we conservatively replace estimates below unity by $\max(\tau(N_\shot), 1)$.
}.
We compare the results obtained with $N_{\shot}=1$ and $N_{\shot}=49$ under the fixed total circuit budget $N_{\tot}=2000$.
The corresponding chain lengths are $N_{\chain}=1000$ for $N_{\shot}=1$ and $N_{\chain}=40$ for $N_{\shot}=49$ because $M=1$ for chiral condensate and \rev{fermion} number density.
Overall, both shot settings reproduce the exact-diagonalization results (black curves) within statistical fluctuations, and the single-shot strategy remains competitive across the parameter range studied.
This trend is particularly clear at the highest temperatures, $\beta g = 0.4$, where the $N_{\shot}=1$ results perform comparably to, or better than, the $N_{\shot}=49$ results, as in the energy-density case.
On the other hand, in the low-temperature regime $\beta g=14.0$, the $N_{\shot}=49$ results sometimes reproduce the exact values more accurately in the small-$\mu$ region, $\mu/g \lesssim 1.5$, especially for the chiral condensate.
This trend is also qualitatively consistent with the general shot-allocation analysis in Sec.~\ref{sec:shot_analysis}, which shows that a larger $N_{\shot}$ can become favorable when the shot noise is relatively more important than the METTS-sampling fluctuations, as the latter decrease at lower temperatures.

\subsection{\rev{Numerical results with QITE}}
\rev{In the previous subsection, we investigated the performance of our
algorithm using exact ITE.
On a quantum computer, however, non-unitary ITE cannot
be implemented directly without additional computational overhead.
Here, we approximate the imaginary-time evolution using the Quantum
Imaginary-Time Evolution (QITE) algorithm proposed by Motta
\textit{et al.}~\cite{Motta:2019yya} and heuristically assess the
performance of our algorithm when combined with QITE.
We first briefly review the QITE algorithm.}


\rev{
In QITE, after Trotter decomposition of ITE, the
action of each local imaginary-time evolution operator is
approximated by a unitary transformation as
\begin{equation}
    \frac{e^{-\Delta\tau h}\ket{\Psi}}
    {\|{e^{-\Delta\tau h}\ket{\Psi}}\|}
    \simeq
    e^{-i\Delta\tau \sum_{I}a_I\sigma_I}\ket{\Psi},
\end{equation}
where $\ket{\Psi}$ denotes the state before the imaginary-time step.
The sum over $I$ runs over the set $\Pauli_D$ of Pauli strings acting on
a domain of $D$ qubits surrounding the support of $h$.
The coefficients $a_I$ are determined by optimizing the norm $\|(\e^{{-i\Delta \tau \sum_I a_I \sigma_I}} - \mathcal{N} e^{-\Delta \tau h}) \ket{\Psi}\|$ on classical computers, where $\mathcal{N} = \|{e^{-\Delta\tau h}\ket{\Psi}}\|^{-1}$.
Since the Pauli basis on a $D$-qubit domain contains $|\Pauli_D|=4^D$ operators, the computational costs increase exponentially with $D$.
For further details, see Ref.~\cite{Motta:2019yya}.
}

\rev{For larger systems, one therefore generally needs to restrict the domain size to fit the available computational resources.
This finite-domain approximation introduces a trade-off between the accuracy of the imaginary-time evolution and its computational cost.
The effectiveness of the domain truncation depends on the correlations between the degrees of freedom inside and outside the chosen domain.
In general quantum many-body systems with local Hamiltonians, correlations of thermal states tend to remain short-ranged in the high-temperature regime~\cite{Kliesch:2014bxj}, which motivates the expectation that finite-domain approximations of the ITE operator can be more effective in this regime.
However, if the gauge degrees of freedom are eliminated by explicitly solving the Gauss-law constraints, the resulting Hamiltonian may become non-local, which can reduce the effectiveness of such approximations.
By contrast, our formulation retains the gauge degrees of freedom and preserves the locality of the Hamiltonian.
Moreover, the geometry-adapted MUPB construction introduced above can be implemented using only local gates. 
These features avoid introducing additional non-locality into either the Hamiltonian or the collapse states and can therefore be expected to be favorable for QITE with truncated domain sizes.}

\begin{figure*}
    \centering
    \includegraphics[width=0.95\linewidth]{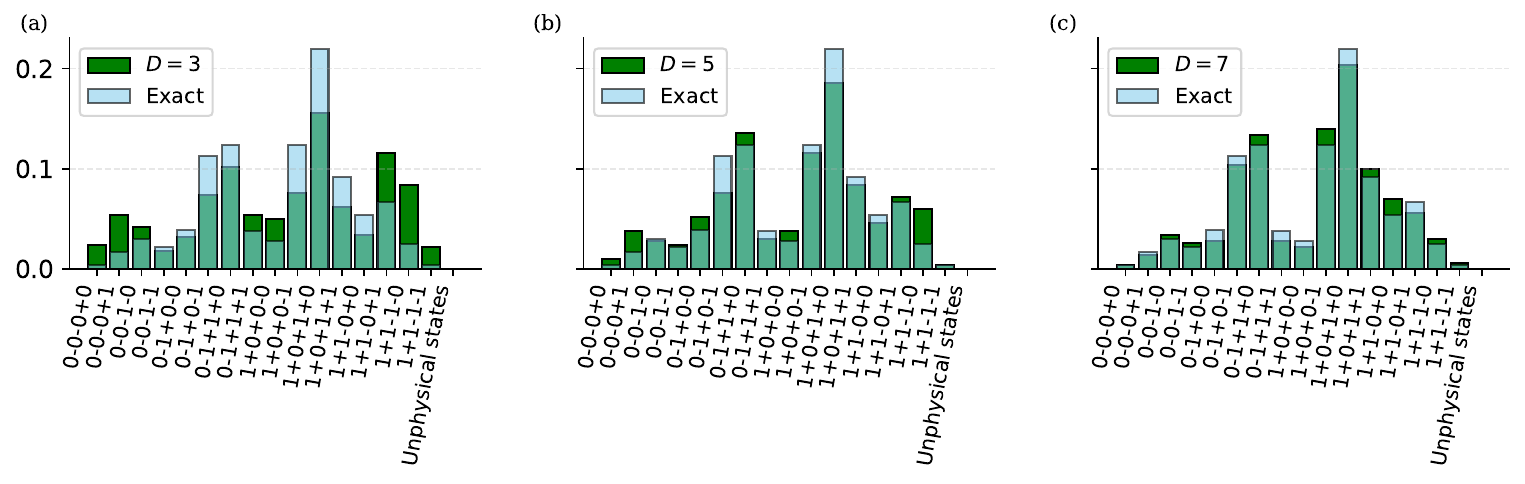}
    \caption{
    \rev{Sampling distributions of collapse states obtained from projective measurements during QMETTS sampling using the MUPB scheme (green), together with the exact probabilities (light blue), for $L_{\mathrm{KS}}=4$ with QITE.
    We plot only the collapse states obtained at the physical $Z$-basis measurement steps of the alternating Markov chain.
    All gauge-violating collapse states are grouped into ``unphysical states''.
    The exact probabilities are calculated as
    $\Prob_i = \langle i^{(Z_\phys)}| e^{-\beta(H-\mu N)} |i^{(Z_\phys)} \rangle / Z$.
    The inverse temperature and chemical potential are fixed at $\beta g=1.0$ and $\mu/g=0.0$, respectively.
    The domain sizes are $D=3,5,$ and $7$ in panels (a), (b), and (c), respectively.
    All QMETTS results are obtained with $N_\chain=1000$.
    For all panels, the initial state is fixed to $\ket{1+0+1+0}$, which corresponds to the trivial no-excitation state in the qubit description introduced in Sec.~\ref{sec:LGTsetup}.}}
    \label{fig:hist_L4_QITE_domain_dependence_mu0_beta1}
\end{figure*}
\begin{figure}
    \centering
    \includegraphics[width=0.9\linewidth]{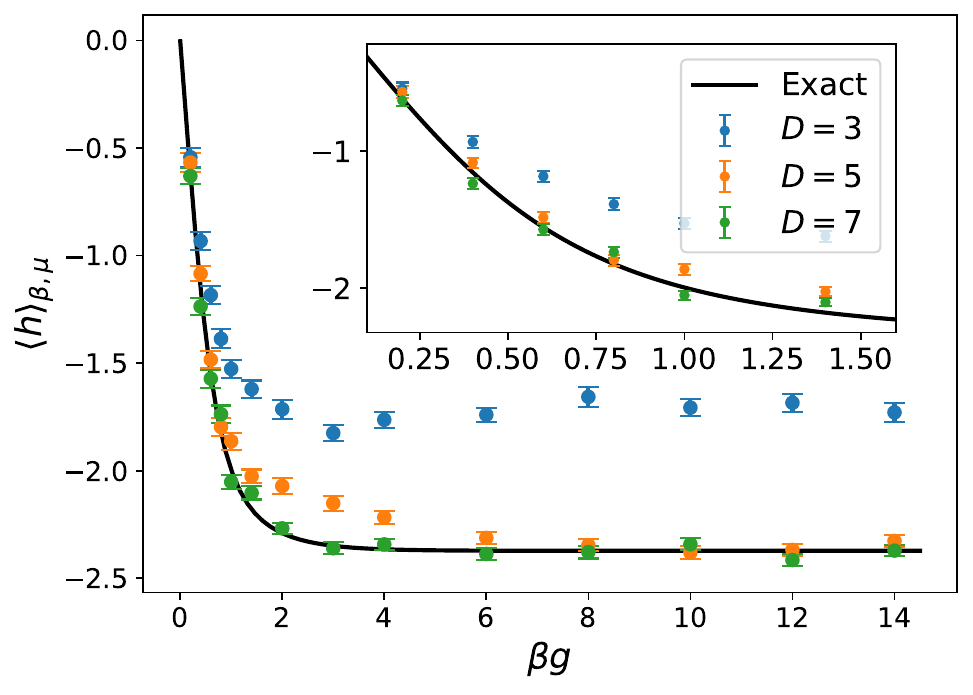}
    \caption{
    \rev{Temperature dependence of the energy density $\braket{h}_{\beta,\mu}$ obtained using QMETTS with MUPB and the single-shot strategy combined with QITE for $L_{\mathrm{KS}}=4$ at zero chemical potential ($\mu/g=0$).
    The results for domain sizes $D=3,5,$ and $7$ are represented by blue, orange, and green dots with error bars, respectively, and are compared with the exact-diagonalization result (black solid line).
    The energy-density measurement requires $M=3$ measurement groups.
    For all domain sizes, the total number of circuit executions is fixed at $N_{\tot}=4000$, corresponding to a Markov-chain length of $N_{\chain}=1000$.
    The inset highlights the high-temperature region $\beta g\in[0,1.6]$.}
    }
    \label{fig:beta_domain_dependence_H_QITE.pdf}
\end{figure}

\rev{
To examine this expectation numerically, we investigate the dependence on the QITE domain size $D$ for $L_{\mathrm{KS}}=4$.
We choose this system size because the full-domain case, $D=N_q=7$, remains computationally tractable, allowing us to directly compare truncated-domain calculations with the full-domain result.
Further details of the QITE implementation in our setup are provided in App.~\ref{app:QITE}.
We first examine how the domain truncation affects the sampling distribution.
In Fig.~\ref{fig:hist_L4_QITE_domain_dependence_mu0_beta1}, we plot the sampling distributions of the collapse states at $\mu/g=0.0$ and $\beta g=1.0$ using the constructed MUPB for domain sizes $D=3,5,$ and $7$.
As shown in Fig.~\ref{fig:hist_L4_QITE_domain_dependence_mu0_beta1},
the QMETTS sampling reproduces the exact stationary distribution more
accurately as the domain size increases.
Nevertheless, even for $D=3$, we observe no gauge-violating collapse
states.
This result suggests that, within the present setup and numerical
precision, introducing a finite-domain cutoff in QITE does not cause
appreciable leakage from the physical sector.
}


\rev{
Next, in Fig.~\ref{fig:beta_domain_dependence_H_QITE.pdf}, we compute the
energy density at $\mu/g=0.0$ by changing domain sizes from $D=3$ to $D=7$.
The results obtained by exact diagonalization are shown by the black curve.
As shown in Fig.~\ref{fig:beta_domain_dependence_H_QITE.pdf}, the full-domain calculation with $D=7$ (green dots) accurately reproduces the exact result, whereas the result for $D=3$ (blue dots) deviates substantially from it.
For $D=5$ (orange dots), the QITE result follows the exact behavior in the high-temperature regime, $\beta g\lesssim 1.0$, as shown in the inset.
These results indicate that, for this model, domain truncation can retain reasonable accuracy in the high-temperature regime while reducing the computational cost as we discussed.
In addition, although deviations appear at intermediate temperatures, the result approaches the exact value again as $\beta g$ becomes large.
This behavior can be interpreted as an effective delay of the cooling process induced by the finite-domain approximation, with the truncated-QITE result eventually approaching the exact low-temperature value at large $\beta$.
}

Finally, we turn to the finite-temperature and finite-density behavior of the model.
To separate these physical features from domain-truncation effects, \rev{we use the full-domain calculation with $D=7$}.
Figure~\ref{fig:phase diagram} summarizes the behavior of the chiral condensate and the \rev{fermion} number density over the $(\mu/g,T/g)$ plane.
The behavior of the chiral condensate is consistent with the expected tendency toward chiral-symmetry restoration as the temperature or chemical potential increases.
The \rev{fermion} number density further indicates that the high-density region appears in the same parameter regime where the chiral condensate is suppressed.
These results demonstrate that the proposed method captures the expected finite-temperature and finite-density structure of the model within a Hamiltonian formalism that does not rely on Euclidean Monte Carlo sampling.
\begin{figure*}[t]
  \begin{minipage}{0.47\textwidth}
    \centering
    \includegraphics[width=1.0\textwidth]{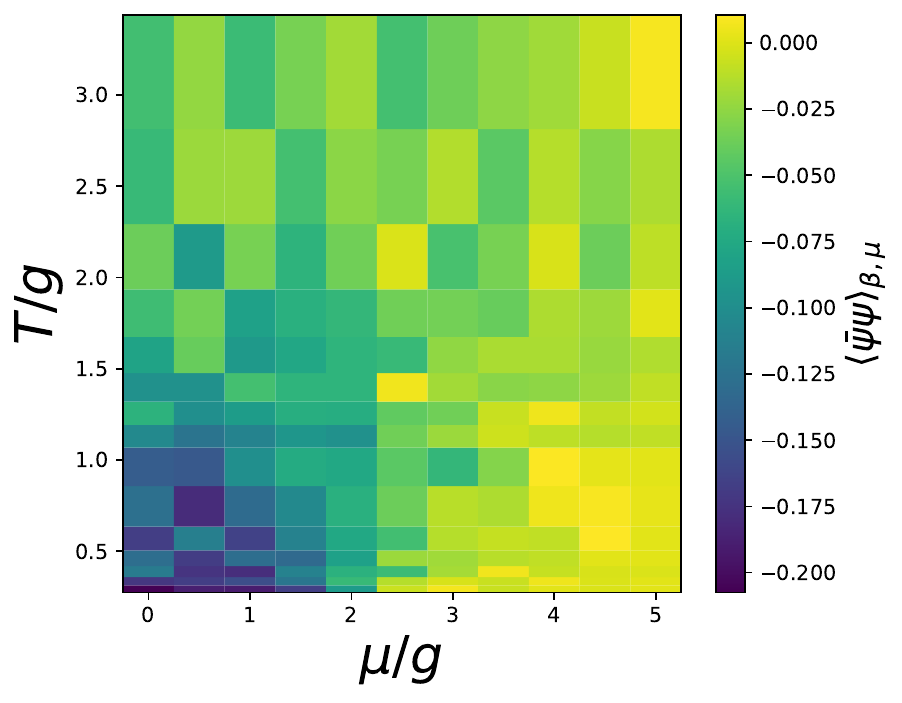}
    \label{fig:phase_diag_chiral_cond}
  \end{minipage}%
  \hfill
  \begin{minipage}{0.47\textwidth}
    \centering
    \includegraphics[width=1.0\textwidth]{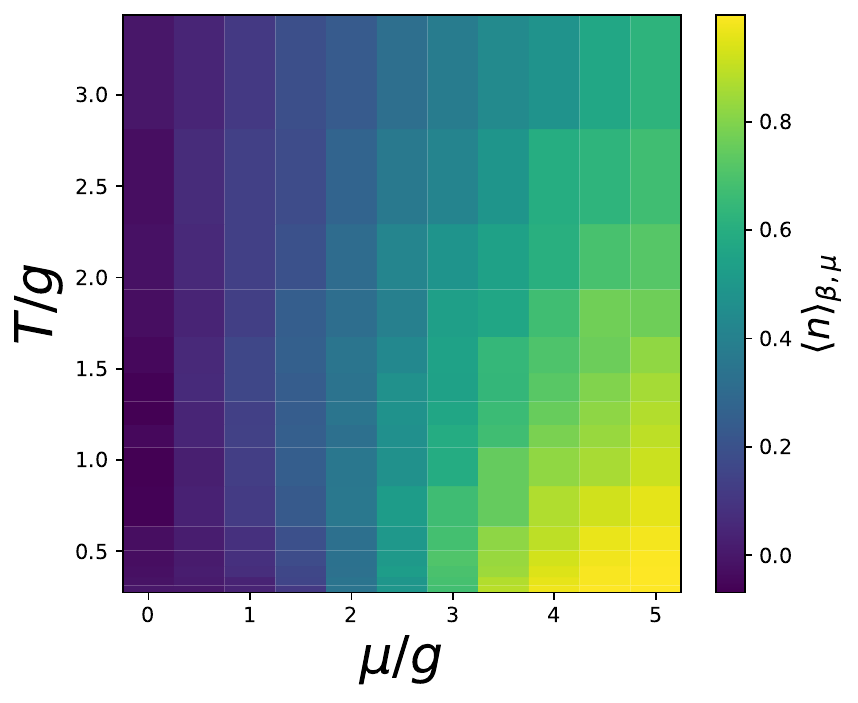}
    \label{fig:phase_diag_number_density}
  \end{minipage}
  \caption{Finite-size phase diagrams in the $(\mu/g,T/g)$ plane for the chiral condensate $\langle \bar{\psi}\psi \rangle_{\beta,\mu}$ (left) and the \rev{fermion} number density $\langle n \rangle_{\beta,\mu}$ (right).
  The system size is $L_{\mathrm{KS}}=4$.
  The QMETTS results are obtained with $N_{\chain}=1000$ using the single-shot strategy.
  Both observables require a single measurement setting, $M=1$.
  The color scales indicate the expectation values estimated by the proposed method. 
  The horizontal axis represents the chemical potential $\mu/g$, and the vertical axis represents the temperature $T/g$.}
  \label{fig:phase diagram}
\end{figure*}

\section{Conclusion} \label{sec:conclusion}
In this work, we have developed a $\mathbb{Z}_2$ gauge-invariant framework for computing finite-temperature and finite-density expectation values by introducing Mutually Unbiased Physical Bases (MUPB) in the QMETTS algorithm. 
Projective measurements within the MUPB ensure that the system never collapses into gauge-violating states and maintain ergodicity within the physical subspace.
By leveraging the mathematical equivalence between $\mathbb{Z}_2$ lattice gauge theory (LGT) and the stabilizer formalism, we have identified an efficient construction for MUPBs using shallow circuits.
This construction is applicable across arbitrary dimensions and boundary conditions.
Furthermore, we analyzed the effect of shot noise on expectation-value estimation in QMETTS.
We showed that the finite-shot QMETTS estimator remains unbiased for finite-temperature and finite-density expectation values, and that, under a fixed total number of circuit executions, the single-shot strategy provides a robust near-optimal choice. 
Its variance is at most a factor of two above the optimum, while requiring no prior knowledge of the relevant variances or autocorrelation times.
\rev{To validate our methodology, we applied the proposed approach to a $(1+1)$-dimensional $\mathbb{Z}_2$ LGT with fermionic matter.
First, using exact ITE, we investigated the performance of our approach and numerically demonstrated that it accurately reproduces thermal properties over a wide range of temperatures and chemical potentials.
We then combined our framework with QITE to investigate its performance in a more practical implementation of imaginary-time evolution.
By varying the QITE domain size, we found that domain truncation can reproduce the exact results reasonably well in the high-temperature regime while reducing the computational cost associated with the ITE step.
While these results demonstrate the potential usefulness of domain truncation in the present finite-size benchmark, they do not determine how the domain size required for a given accuracy scales with the system size.
Establishing this scaling for larger systems remains an important direction for future work.}

\rev{
Our stabilizer-based construction may also be applicable to a broader class of gauge theories that admit stabilizer formulations.
First, $\mathbb{Z}_d$ lattice gauge theories with prime $d$ possess an analogous stabilizer structure~\cite{Gottesman:1998se,Gheorghiu:2011awf,Spagnoli:2026qni}.
For $\mathbb{Z}_d$ lattice gauge theories with $d$ equal to a power of two, it has also been shown that the theory can be reduced to a stabilizer formulation~\cite{Turco:2026cte}.
For continuous gauge groups, certain truncated formulations of $\mathrm{U}(1)$ and $\mathrm{SU}(2)$ lattice gauge theories can likewise be mapped to qubit- or qudit-based stabilizer descriptions ~\cite{Stryker:2018efp, Rajput:2021trn, Lacambra:2026twp, Yao:2025cxs}, although the applicability depends on the truncation scheme and boundary conditions.
Importantly, our MUPB construction relies on the stabilizer formalism rather than on properties specific to $\mathbb{Z}_2$ gauge theories.
Therefore, whenever a gauge theory admits a compatible stabilizer formulation, our construction can be extended to that setting straightforwardly without modifying its basic structure.
Extending our algorithm to more general continuous gauge groups remains a promising direction for future work.}

Our main contribution is independent of the particular implementation of ITE, although the practical efficiency of the overall framework depends on the underlying ITE scheme.
In the present work, we employ the QITE algorithm proposed by Motta \textit{et al.} and heuristically investigate its performance by varying the domain size.
As the domain size increases, the accuracy improves, while the computational cost grows exponentially.
Recent studies have shown, however, that exploiting gauge symmetries can substantially reduce the cost of the ITE component in gauge-theory simulations~\cite{Sekiyama:2026rgt}.
\rev{In addition, combining our sampling framework with alternative ITE schemes, such as probabilistic implementations of ITE
~\cite{Kosugi:2021kgg,Turro:2021vbk,Gilyen:2018khw}, adiabatic QITE~\cite{Hejazi:2023fiq}, and double-bracket QITE~\cite{Gluza:2024lqq}, provides an important direction for future work.
It would also be interesting to explore connections with sampling-based finite-temperature algorithms that use series expansions of imaginary-time evolution~\cite{Matsumoto:2024won}.
}

Finally, our proposed framework is naturally compatible with emerging techniques that leverage gauge symmetry for symmetry verification~\cite{Schmale:2024ebh, Ballini:2024qmr, Rajput:2021trn, Carena:2024dzu, Spagnoli:2024mib}. 
In such error-mitigation schemes, quantum errors are detected by monitoring violations of gauge constraints, a process that inherently relies on the redundancy of the Hilbert space.
Unlike standard approaches that may trace out or eliminate these degrees of freedom, our method explicitly preserves this redundancy, allowing for straightforward integration with symmetry-based error suppression.
By combining our MUPB-based QMETTS with these verification protocols, our approach provides a promising route toward more robust quantum simulations of lattice gauge theories.

\begin{acknowledgments}
We thank Etsuko Itou, Koji Terashi, and Lento Nagano for fruitful discussions and/or comments.
We thank Yutaro Iiyama for assistance with numerical calculations.
This work was supported by JST SPRING, Grant Number JPMJSP2108.
\end{acknowledgments}

\appendix

\section{Efficient construction of MUPB for $\mathbb{Z}_2$ LGTs in general dimensions and with arbitrary boundary conditions} \label{app:any_MUPB}
In this section, we provide a formal proof that the quantum circuits illustrated in Fig.~\ref{fig:gen_MUPB} satisfy the conditions for Mutually Unbiased Physical Bases (MUPB) defined through Def.~\ref{def:MUPB}.
Our proof leverages the stabilizer formalism, which is particularly well-suited for describing subspaces defined by symmetry constraints, such as gauge invariance.
First, we provide a brief review of the stabilizer formalism.
Subsequently, we apply this framework to the $\mathbb{Z}_2$ lattice gauge theory to rigorously prove our statement.

\subsection{Stabilizer formalism} \label{subsec:stab_formalism}
Let $G$ be a group.
If a subset $H \subseteq G$ can generate all elements of $G$ through the multiplication of its elements, we refer to $H$ as a generating set of $G$ and denote the group as $G = \langle H \rangle$.
The elements of $H$ are called the generators of $G$.
The $N_q$-qubit Pauli group is defined as, 
$
    \Pauli_{N_q} := \{\pm 1, \pm i\} \otimes \{I , X, Y, Z\}^{\otimes N_q}
$.
Stabilizer groups, operators, and states are defined using the Pauli group as follows.
\begin{definition}(Stabilizer groups, operators, and states)\\
    A subgroup $\Stab_{N_q} \subseteq \mathcal{P}_{N_q}$ is called a stabilizer group if it is Abelian (all elements commute) and does not contain the element $-I$.
    The elements $S \in \Stab_{N_q}$ are referred to as stabilizer operators.
    A quantum state $\ket{\psi}$ is called a stabilizer state of $\Stab_{N_q}$ if it satisfies:
    \begin{align}
        S_i \ket{\psi} = +\ket{\psi} \quad \text{for } \forall S_i \in \Stab_{N_q}. \nonumber
    \end{align}
\end{definition}

Note that all stabilizer operators are Hermitian and the eigenvalues must be $\pm 1$ since $-I \notin \Stab_{N_q}$.
The framework for describing quantum states and subspaces using these groups is known as the stabilizer formalism.
The theory of quantum error correction is mainly based on this formalism; logical qubits are embedded into the space spanned by stabilizer states called code space, and errors can be detected by measuring the eigenvalues of quantum states.

We also define the $N_q$-qubit Clifford group, which plays an important role in describing stabilizer formalism.
$N_q$-qubit Clifford group $\Clifford_{N_q}$ is defined as the normalizer of the Pauli group, 
\begin{align}
    \Clifford_{N_q} := \{C \in \Unitary(\Complex^2)^{\otimes {N_q}}| P \in \Pauli_{N_q}, CPC^{\dagger} \in  \Pauli_{N_q}\}. \nonumber
\end{align}
The Clifford group consists of Hadamard gates, Phase gates, and CNOT gates.
We now summarize some key properties of the stabilizer formalism to provide our proof.
\begin{theorem}[\cite{Aaronson:2004xuh}] \label{th:code_dim}
    Let the system size be $N_q$ and the number of generators of a stabilizer group $\Stab_{N_q}$ be $S$.
    The dimension of the space spanned by the stabilizer states is given by
    \begin{align}
        2^{N_q} / 2^S = 2^{N_q - S}. \nonumber
    \end{align}
\end{theorem}
\begin{proof}
    For any non-trivial stabilizer operator $S_i \in \mathcal{S}_{N_q} \setminus \{I\}$, its eigenvalues are $\pm 1$.
    Since $\Tr[S_i] = 0$ for any Pauli string other than the identity, the multiplicities of the $+1$ and $-1$ eigenvalues must be equal.
    Consequently, the projection onto the $+1$-eigenspace of a single generator $S_i$ bisects the Hilbert space, reducing its dimension by half.
    Given $S$ independent and commuting generators, each additional generator further constrains the space by a factor of $1/2$.
    Therefore, the total dimension of the simultaneous $+1$-eigenspace is $2^{N_q-S}$.
\end{proof}

Next, we discuss the inner product of two stabilizer states $\ket{\phi}$ and $\ket{\psi}$ with different stabilizer groups.
\begin{theorem}[\cite{Aaronson:2004xuh}] \label{th:inner_prod}
    Let $\ket{\psi}$ and $\ket{\phi}$ be stabilizer states in an $N_q$-qubit system, uniquely determined (up to a global phase) by their respective stabilizer groups $\Stab(\ket{\psi}) = \langle \{P_1, \dots, P_{N_q} \}\rangle$ and $\Stab(\ket{\phi}) = \langle\{ Q_1, \dots, Q_{N_q}\} \rangle$. The magnitude of their inner product $|\braket{\phi|\psi}|$ is determined as follows,
    \begin{enumerate}
        \item If there exists an operator $P \in \mathcal{P}_{N_q}$ such that $P \in \Stab(\ket{\psi})$ and $-P \in \Stab(\ket{\phi})$, the states are orthogonal: $|\braket{\phi|\psi}| = 0$.
        \item 
        Let $S$ be the minimum, over all sets of generators $\{P_1, \dots , P_{N_q}\}$ for $\Stab(\ket{\psi}) = \langle \{P_1, \dots, P_{N_q} \}\rangle$ and $\{Q_1, \dots , Q_{N_q}\}$ for $\Stab(\ket{\phi}) = \langle\{ Q_1, \dots, Q_{N_q}\} \rangle$, of the number of different $i$ values for which $P_i \neq Q_i$. 
        Then, $|\braket{\phi|\psi}| = 2^{-S/2}$.
    \end{enumerate}
\end{theorem}
Note that the number of different generators cannot be uniquely determined and we must maximize the number of the same generators in step 2.
\begin{proof}
    First, we prove the first step.
    If there exists $P$ such that $P\ket{\psi} = \ket{\psi}$ and $-P\ket{\phi} = \ket{\phi}$,
    $\ket{\psi}$ and $\ket{\phi}$ are in different eigenspaces.
    Thus, we obtain $\braket{\phi|\psi} = 0$.

    Next, we prove the second step.
    By the theorem~\ref{th:decomp_stab}, there always exists an element of the Clifford group $W$ such that $\ket{\psi}$ transforms into canonical form $W\ket{\psi} = \ket{0}^{\otimes N_q}$, indicating that $W\ket{\psi}$ is stabilized by $\langle\{Z_1, \dots, Z_{N_q}\}\rangle$.
    Now choose generators of $\Stab(W\ket{\phi})$ so that the number of generators shared with $\Stab(W\ket{\psi})$ is maximized.
    Let $R$ be the number of such common independent generators.
    By definition, the minimum number of different generators is therefore $S=N_q-R$.
    Because $\ket{\psi}$ and $\ket{\phi}$ are assumed to be non-orthogonal, $\Stab(W\ket{\phi})$ cannot contain any operator of the form $-P$ with $P\in \Stab(W\ket{\psi})$; otherwise, $W\ket{\psi}$ and $W\ket{\phi}$ would belong to different eigenspaces of $P$ and hence be orthogonal.
    The remaining $N_q-R$ generators necessarily contain at least one $X$ or $Y$ literal.
    Each one relates pairs of computational-basis amplitudes up to a phase and hence contributes a factor $1/\sqrt{2}$ to the overlap.
    Therefore, 
    $
    |\braket{\phi|\psi}|
    =
    |\braket{W\phi|W\psi}|
    =
    |\braket{W\phi|0^{\otimes N_q}}|
    =
    2^{-S/2}$.
\end{proof}

Next,  we introduce the binary expression of the stabilizer formalism to explain the canonical form of stabilizer groups~\cite{Aaronson:2004xuh}.
Given $N_q$-qubit systems, if we have $S$ generators, stabilizers are represented by $S \times 2N_q$ matrices, where all elements are $0$ or $1$, that is $\mathbb{F}_2^{S \times 2N_q}$.
The rows correspond to each stabilizer and the columns to each qubit.
The former part of the matrices represents $X$ and the latter $Z$.
Global phases get dropped.
For example, the five-qubit code \cite{nielsen2010quantum}:
\begin{align}
    XZZXI,
    IXZZX,
    XIXZZ,
    ZXIXZ , \nonumber
\end{align}
can be equivalently expressed in the following matrix up to global phases:
\begin{align}
\left(
  \begin{array}{ccccc|ccccc}
    1 & 0 & 0 & 1 & 0 & 0 & 1 & 1 & 0 & 0 \\
    0 & 1 & 0 & 0 & 1 & 0 & 0 & 1 & 1 & 0 \\
    1 & 0 & 1 & 0 & 0 & 0 & 0 & 0 & 1 & 1 \\
    0 & 1 & 0 & 1 & 0 & 1 & 0 & 0 & 0 & 1 
  \end{array}
\right). \nonumber
\end{align}
In this expression, $XZZXI$ corresponds to the first row $(1  0  0  1  0 ~ | ~  0  1  1  0  0)$.
We denote the left (right) space as $X$- ($Z$-) space.
In this expression, Clifford gates (Hadamard, Phase, and CNOT) can be interpreted as elementary column operations.
Since $\mathsf{H}X\mathsf{H} = Z, \mathsf{H}Y\mathsf{H} = -Y$,  Hadamard gates play a role in swapping corresponding two columns.
$\mathsf{P}Z\mathsf{P}^\dagger = Z, \mathsf{P}X\mathsf{P}^\dagger = Y$ means that Phase gates $\mathsf{P}$ add the columns of $X$ to the columns of $Z$ modulo $2$.
CNOT$_{i, j}$, which controls $i$-th qubit and acts $X$ gate to $j$-th qubit, is equivalent to $x_j \rightarrow x_j ~\oplus ~x_i$ and $z_i \rightarrow z_i ~\oplus ~z_j$ since $\CNOT_{i, j} X_i \CNOT_{i, j}^\dagger = X_i X_j$ and $\CNOT_{i, j} Z_j \CNOT_{i, j}^\dagger = Z_i Z_j$, where $x_i (z_i)$ is the $i$-th column vector and $\oplus$ represents addition modulo $2$.
Since generators are independent, we can make the $X$-space block have rank $S$ by using $\mathsf{H}$ gates.
Due to the nature of $\mathbb{F}_2$, $\CNOT$ gates play a role in adding and swapping column vectors in $X$-space, therefore we can perform Gaussian elimination in $X$-space by using CNOT gates.
In addition, the commutativity condition imposes a constraint on the matrix.
If the space is represented by $(A|B)$, the commutativity condition requires $AB^T$ to be symmetric.

Based on this expression, we prove the following theorem which is the modified version of Theorem~8 in the literature~\cite{Aaronson:2004xuh}.
\begin{theorem} \label{th:decomp_stab}
    Given the binary representation of \(S\) independent commuting stabilizer generators in \(\mathbb{F}_2^{S\times 2N_q}\), we can transform them into the canonical form by the application of Clifford gates:
    \begin{align}
        \left(
          \begin{array}{cc|cc}
            0 & 0  & I & 0
          \end{array}
        \right), \nonumber
    \end{align}
    where $I \in \mathbb{F}_2^{S \times S}$.
\end{theorem}
To prove the above theorem, we show a lemma below.
\begin{lemma}[\cite{Aaronson:2004xuh}] \label{lem:sym}
    Let $k$ be an integer $k \in \mathbb{N}$.
    For any symmetric matrix $A \in \mathbb{Z}_2^{k \times k}$, there exists a diagonal matrix $\Lambda$ such that $A + \Lambda = MM^T$, with $M$ some invertible binary matrix.
\end{lemma}
\begin{proof}
    Let $M$ be a lower-triangular matrix with $1$s all along the diagonal:
    \begin{align}
        M_{ii} &= 1, \nonumber \\
        M_{ij} &= 0 ~ (i < j). \nonumber
    \end{align}
    Such an $M$ is always invertible.
    If we rewrite the claim in terms of elements,
    \begin{align}
        A_{ij} + \Lambda_{ij}= \sum_k M_{ik} M_{jk}, \label{eq:tri_ang}
    \end{align}
    for all pairs $(i,j)$.
    In the case $i=j$, it always holds by taking $\Lambda$ appropriately.
    Thus, considering the case $i>j$ is enough since Eq.~(\ref{eq:tri_ang}) is symmetric under $i$ and $j$.
    Hence, we only consider the case $i>j$ below.

    Since $M$ is a lower-triangular matrix, $M_{ik}M_{jk}=0$ unless $k \le j$, following 
    \begin{align}
        A_{ij} = \sum_{k<j} M_{ik} M_{jk} + M_{ij}. \label{eq:core_lemma}
    \end{align}
    Based on those observations, we prove we can determine all elements of $M$ by mathematical induction about $j$.
    
    For $j=1$, since there is no $k$ such that $k<j=1$, Eq.~(\ref{eq:core_lemma}) indicates that $A_{i1} = M_{i1}$ for $\forall ~i~(>1)$.
    As $M_{11} = 1$, we can determine the elements of $M_{i1}$ for $ \forall i$.
    Next, we assume that $M_{ij'}$ have already been determined for $\forall ~i$ and $j'<j$.
    Then, we can obtain $M_{ij}$ from Eq.~(\ref{eq:core_lemma}) because
    \begin{align}
        M_{ij} = A_{ij} - \sum_{k<j} M_{ik} M_{jk} \nonumber
    \end{align}
    and RHS is already determined from the hypothesis for $j'<j$.
\end{proof}

Using the above lemma, let us derive the theorem.
\begin{proof}
    Let \((A|B)\in\mathbb{F}_2^{S\times 2N_q}\) be such a binary representation, with \(A,B\in\mathbb{F}_2^{S\times N_q}\).
    Consider the following steps.
    \begin{enumerate}
        \item Use Hadamard gates to make $A$ have rank $S$.
        \item Use $\CNOT$s to perform Gaussian elimination on $A$, producing
        \begin{align}
            \left(
              \begin{array}{cc|cc}
                I & 0  & C & D
              \end{array}
            \right), \nonumber
        \end{align}
        where $C \in \mathbb{F}_2^{S \times S}$ and $D \in \mathbb{F}_2^{S \times (N_q-S)}$ are the former and latter parts of $B$ after Gaussian elimination on $X$-space.
        \item Use Hadamard gates to swap,
        \begin{align}
            \left(
              \begin{array}{cc|cc}
                I & D  & C & 0
              \end{array}
            \right). \nonumber
        \end{align}
        \item Again, use $\CNOT$s to perform Gaussian elimination on $X$-space, producing
        \begin{align}
            \left(
              \begin{array}{cc|cc}
                I & 0  & C & 0
              \end{array}
            \right). \nonumber
        \end{align}
        \item Commutativity states that $IC^T$ is symmetric, implying that $C$ is symmetric.
        We can apply phase gates to add a diagonal matrix to $C$ to convert $C$ into the form $C = MM^T$ for some invertible $M$ based on Lemma~\ref{lem:sym}.
        \item Use $\CNOT$s, producing
        \begin{align}
            \left(
              \begin{array}{cc|cc}
                M & 0  & M & 0
              \end{array}
            \right). \nonumber
        \end{align}
        \item Apply phase gates to $S$ qubits to obtain
        \begin{align}
            \left(
              \begin{array}{cc|c}
                M & 0  & 0 
              \end{array}
            \right). \nonumber
        \end{align}
        \item Again, use $\CNOT$s to perform Gaussian elimination on $M$, producing
        \begin{align}
            \left(
              \begin{array}{cc|c}
                I & 0  & 0 
              \end{array}
            \right). \nonumber
        \end{align}
        \item Use Hadamard gates for swapping $X$- and $Z$-space
        \begin{align}
            \left(
              \begin{array}{c|cc}
                0 & I  &  0
              \end{array}
            \right). \nonumber
        \end{align}
    \end{enumerate}
\end{proof}
Let us evaluate the circuit complexity in this algorithm.
The parallel application of Hadamard gates and Phase gates implies that steps 1, 3, 5, 7, and 9 can be achieved with constant depth and $\mathcal{O} (N_q)$ gates.
Steps 2, 4, 6, and 8 involving $\CNOT$s can be constructed by $\mathcal{O}(N_q^2)$ gates with $\mathcal{O}(\log N_q)$ depth for each step based on \cite{Moore:1998hz}.
In total, we can transform any stabilizer into canonical form using $\mathcal{O}(N_q^2)$ gates with $\mathcal{O}(\log N_q)$ depth.

\subsection{Efficient construction of the MUPB}
Here, we prove that measurement bases shown in Fig.~\ref{fig:gen_MUPB} satisfy the MUPB condition (Def.~\ref{def:MUPB}) in arbitrary spatial dimensions and boundary conditions and they can be constructed efficiently, i.e. $\mathcal{O}(N_q^2)$ gates and  $\mathcal{O}(\log N_q)$ depth for the system size $N_q$ based on the stabilizer formalism.

For $\mathbb{Z}_2$ LGTs in arbitrary dimensions and with general boundary conditions, the gauge groups spanned by independent Gauss's law operators~\footnote{
As follows from Theorem~\ref{th:decomp_stab}, the stabilizer generators corresponding to Gauss's law constraints are assumed to be independent. 
If the original set of Gauss's law operators contains redundant constraints, we first replace it with an independent generating subset.
} $\mathcal{G}_{N_q} = \braket{\{G_n\}_{n=1}^S}$ can be efficiently decomposed into the trivial stabilizer group using Theorem~\ref{th:decomp_stab} followed by the Corollary below.
\begin{corollary} \label{cor:decompse_Z2}
    For any $\mathbb{Z}_2$ gauge groups $\mathcal{G}_{N_q}$ with $S$ independent generators, there exists an element of Clifford group $W \in \Clifford_{N_q}$ such that
    \begin{align}
        W \mathcal{G}_{N_q} W^{\dagger} 
        &= \{I, Z\}^{\otimes S} \otimes \{I\}^{\otimes N_q - S} \nonumber \\
        &= \braket{\{Z_1, Z_2, ..., Z_S\}}. \label{eq:decomp}
    \end{align}
    $W$ can be efficiently computed classically and constructed using $\mathcal{O}(N_q^2)$ Hadamard gates, Phase gates, and CNOT gates with $\mathcal{O}(\log N_q)$ depth.
\end{corollary}

Based on this corollary, we can derive the efficient quantum circuits to
construct the MUPB.
\begin{theorem}
    The circuits in Fig.~\ref{fig:gen_MUPB} satisfy the conditions of the MUPB.
    $W$ is the element of the Clifford group shown in Corollary~\ref{cor:decompse_Z2}.
    $W$ can be constructed using $\mathcal{O}(N_q^2)$ Hadamard gates, Phase gates and CNOT gates with $\mathcal{O}(\log N_q)$ depth.
\end{theorem}
Here, we denote the measurement bases represented in Figs.~\ref{fig:gen_MUPB} by the physical $Z$-basis and the physical $X$-basis with their respective basis states denoted by $\{\ket{e_i^{(Z_{\text{phys}})}}\}$ and $\{\ket{e_i^{(X_{\text{phys}})}}\}$ again.
\begin{proof}
    First, we show that the two measurement bases are eigenstates of Gauss's law operators.
    The measurement basis in Fig.~\ref{fig:gen_MUPB} (a) is given by
    $
        \ket{e_i^{(Z_\phys)}} =  W^{\dagger}\ket{b_i}
    $
    in terms of bit string states $\ket{b_i} = \ket{b_i(1) b_i(2) ... b_i({N_q})}$, where $b_i(j) \in \{0, 1\}$ means the $j$-th bit string of the state $\ket{b_i}$, $\{\ket{e_i^{(Z_\phys)}}\}_i$.
    Therefore, 
    \begin{align}
        G_n \ket{e_i^{(Z_\phys)}} 
        &= W^{\dagger} (W G_n W^{\dagger}) \ket{b_i} \nonumber\\
        &= z W^{\dagger} \ket{b_i} \nonumber\\
        &= z \ket{e_i^{(Z_\phys)}}, \nonumber
    \end{align}
    where $z$ is the corresponding eigenvalue.
    Note that $WG_nW^{\dagger} \in \{I, Z\}^{\otimes S} \otimes \{I\}^{N_q-S}$ and $\ket{b_i}$ is one of the eigenvectors of $WG_nW^{\dagger}$.
    Similarly, we can show
    \begin{align}
        G_n \ket{e_i^{(X_\phys)}} 
        &= z \ket{e_i^{(X_\phys)}}. \nonumber
    \end{align}

    Finally, we show the mutually unbiasedness of two measurement bases.
    By Theorem~\ref{th:code_dim} , $d_{\phys} = 2^{N_q-S}$.
    A measurement basis state $\ket{e_i^{(Z_\phys)}}$ is uniquely stabilized by
    \begin{align}
        \braket{\{G_n\}_{n=1}^S \cup \{W^{\dagger} (-1)^{b_i(n)}Z_n W\}_{n=S+1}^{N_q}}. \nonumber   
    \end{align}
    On the other hand, a measurement basis $\ket{e_i^{(X_\phys)}}$ is uniquely stabilized by
    \begin{align}
        \braket{\{G_n\}_{n=1}^S \cup \{W^{\dagger}(-1)^{b_i(n)} X_n W\}_{n=S+1}^{N_q}}.    \nonumber
    \end{align}
    Therefore, the number of different stabilizers is $N_q-S$.
    Thus,
    \begin{align}
        |\braket{e_i^{(X_\phys)}|e_j^{(Z_\phys)}}| = \frac{1}{2^{(N_q-S)/2}} = \frac{1}{\sqrt{d_{\phys}}}, \nonumber
    \end{align}
    implying that two measurement bases satisfy the condition of the MUPB.
\end{proof}

\section{Autocorrelation analysis} \label{app:autocorr}
In this appendix, we first derive Eq.~(\ref{eq:opt_shot}).
Next, we show Eq.~(\ref{eq:single_shot_ineq}).
Finally, we compare the single-shot strategy with direct independent measurements on the Gibbs state.

First, for $t\ge 1$, the shot noise is independent of the Markov chain, and therefore
\begin{align}
    \Cov(O_k,O_{k+t})=\Cov(\mu_{k},\mu_{k+t}). \label{eq:cov}
\end{align}
Using the definition of the integrated autocorrelation time and Eq.~(\ref{eq:cov}), we obtain
\begin{align}
    \tau(N_\shot) 
    &= 1 + 2 \sum_t \rho_O(t) \nonumber \\
    &= 1 + 2 \sum_t \frac{\Cov(O_k, O_{k+t})}{\sigma_\mu^2 + (1/N_\shot) \sigma_\shot^2} \nonumber \\
    &= 1 + \frac{\sigma_\mu^2}{\sigma_\mu^2 + (1/N_\shot) \sigma_\shot^2} \left(2\sum_t \rho_\mu(t)\right). \nonumber
\end{align}
Since $\tau_\mu = 1 + 2\sum_{t=1}^\infty \rho_\mu(t)$, we find
\begin{align}
    \tau(N_\shot) = 1 + \frac{N_\shot}{N_\shot + r} (\tau_\mu - 1), \label{eq:auto_Nshot}
\end{align}
where we define $r:= \sigma_\shot^2/\sigma_\mu^2$.
The variance is then given by
\begin{align}
    &\Var[\bar O](N_\shot) \nonumber \\
    &= \frac{\tau(N_\shot)}{N_\chain} \sigma_\mu^2 + \frac{\tau(N_\shot)}{N_\shot N_\chain} \sigma_\shot^2 \nonumber\\
    &= \frac{\tau(N_\shot)(N_\shot M + 1)}{N_\tot} \left( \sigma_\mu^2 + \frac{\sigma_\shot^2}{N_\shot}  \right) \nonumber\\
    &= \frac{\sigma_\mu^2}{N_\tot} \frac{(N_\shot M+1)(N_\shot \tau_\mu + r)}{N_\shot}. \nonumber
\end{align}
Treating $N_\shot$ as a continuous variable, the stationary point of this function is given by
\begin{align}
    N_\shot^* = \sqrt{\frac{r}{\tau_\mu M}} = \sqrt{\frac{\sigma_\shot^2}{\sigma_\mu^2 \tau_\mu M}}. \nonumber
\end{align}
Because we require $N_\shot\ge 1$, the optimal shot number is $N_\shot^* \sim \max(1, \sqrt{{\sigma_\shot^2}/{\sigma_\mu^2 \tau_\mu M}})$.

Next, we prove Eq.~(\ref{eq:single_shot_ineq}).
When $r/\tau_\mu \le M$, the constrained optimum is simply $N_\shot^*=1$, and hence
\begin{align}
    \frac{\Var[\bar O](1)}{\Var[\bar O](N_\shot^*)}=1. \label{eq:ratio_ineq_2}
\end{align}
When $r/\tau_\mu > M$, the continuous optimum satisfies $N_\shot> 1$, and thus
\begin{align}
    \Var[\bar O](N_\shot^*) 
    &= \frac{\sigma_\mu^2}{N_\tot} \frac{(M\sqrt{\frac{r}{\tau_\mu M}} + 1)(\tau_\mu \sqrt{\frac{r}{\tau_\mu M}} + r)}{\sqrt{\frac{r}{\tau_\mu M}}} \nonumber\\
    &= \frac{\sigma_\mu^2}{N_\tot} (\sqrt{r M} + \sqrt{\tau_\mu})^2. \nonumber
\end{align}
On the other hand, 
\begin{align}
    \Var[\bar O](N_\shot=1) = \frac{\sigma_\mu^2}{N_\tot} (M+1) (\tau_\mu + r). \nonumber
\end{align}
Therefore, the ratio becomes
\begin{align}
    R:= \frac{\Var[\bar O](N_\shot=1)}{\Var[\bar O](N_\shot^*)}=\frac{(M+1)(\tau_\mu + r)}{(\sqrt{\tau_\mu} + \sqrt{r M})^2}. \nonumber
\end{align}
If we set $x := \sqrt{r / \tau_\mu}$, then $R = [(M+1)(1+x^2)]/(1+\sqrt{M} x)^2 ~(\sqrt{M} < x <\infty)$.
Due to $dR(x)/dx > 0$, $R(x)$ is monotonically increasing function.
Considering
\begin{align}
    \lim_{x \rightarrow \sqrt{M}} R(x) = 1,\nonumber
\end{align}
and 
\begin{align}
    \lim_{x\rightarrow \infty} R(x) = \frac{M+1}{M} = 1 + \frac{1}{M}, \nonumber
\end{align}
we conclude
\begin{align}
    1 < R(x) < 1 + \frac{1}{M}. \label{eq:ratio_ineq_1}
\end{align}
Combining the two cases (Eq.~(\ref{eq:ratio_ineq_2}) and Eq.~(\ref{eq:ratio_ineq_1})), we conclude that
\begin{align}
    1 \le \frac{\Var[\bar O](1)}{\Var[\bar O](N_\shot^*)} < 1+ \frac{1}{M} \le  2. \nonumber
\end{align}
This proves Eq.~(\ref{eq:single_shot_ineq}).

Finally, we compare the single-shot strategy with direct independent
measurements on the Gibbs state. As in the QMETTS setting considered above,
we assume that the observable is decomposed as
$
    \Obs = \sum_{m=1}^M \Obs_m ,
$
and that each component $\Obs_m$ can be measured with one quantum circuit on
the Gibbs state
$
    \rho_{\beta, \mu} = {e^{-\beta (H - \mu N)}}/{Z}.
$
For a fixed total number of circuit executions $N_\tot$, the number of shots
allocated to each component $\Obs_m$ is
$
    N_\shot^\Gibbs = {N_\tot}/{M}.
$
The corresponding estimator is
\begin{align}
    \bar O_\Gibbs
    =
    \frac{1}{N_\shot^\Gibbs}
    \sum_{l=1}^{N_\shot^\Gibbs}
    \sum_{m=1}^M Y_{m,l}, \nonumber 
\end{align}
where $Y_{m,l}$ denotes the $l$-th measurement outcome of $\Obs_m$ on $\rho_{\beta, \mu}$.
Assuming that the measurements for different $m$ and $l$ are
independent, its variance is
\begin{align}
    \Var_\Gibbs[\bar O_\Gibbs]
    &=
    \frac{1}{N_\shot^\Gibbs}
    \sum_{m=1}^M
    \Var_\Gibbs[Y_{m,l}]  \nonumber \\
    &=
    \frac{1}{N_\shot^\Gibbs}
    \sigma_\Gibbs^2 , \nonumber
\end{align}
where
\begin{align}
    \sigma_\Gibbs^2
    :=
    \sum_{m=1}^M
    \left[
    \Tr[\Obs_m^2 \rho_{\beta, \mu}]
    -
    \left(\Tr[\Obs_m \rho_{\beta, \mu}]\right)^2
    \right]. \nonumber 
\end{align}

We next relate this quantity to the variance components appearing in the
QMETTS estimator. 
Let $\mu_{m,i} := \bra{\phi_i}\Obs_m\ket{\phi_i}$, implying that $\mu_i := \sum_{m=1}^M \mu_{m,i}$.
Using the METTS ensemble identity $\sum_i \Prob_i \bra{\phi_i} A \ket{\phi_i}=\Tr[A\rho_{\beta, \mu}]$, for any observable $A$, we obtain
\begin{align}
    \sigma_\shot^2 + \sigma_\mu^2
    &=
    \sum_i \Prob_i
    \sum_{m=1}^M
    \left[
    \bra{\phi_i}\Obs_m^2\ket{\phi_i}
    -
    \mu_{m,i}^2
    \right] \nonumber \\
    &\quad+
    \Var_\chain[\mu_i] \nonumber \\
    &=
    \sum_{m=1}^M
    \left[
    \Tr[\Obs_m^2\rho_{\beta, \mu}]
    -
    \left(\Tr[\Obs_m\rho_{\beta, \mu}]\right)^2
    \right] \nonumber \\
    &\quad +
    2\sum_{m<m'}
    \Cov_\chain(\mu_{m,i},\mu_{m',i}) \nonumber \\
    &=
    \sigma_\Gibbs^2
    +
    2\sum_{m<m'}
    \Cov_\chain(\mu_{m,i},\mu_{m',i}) .\nonumber
\end{align}
The covariance term 
\begin{align}
    &\Cov_\chain(\mu_{m,i},\mu_{m',i}) \nonumber \\
    &= \sum_{i} \Prob_i \mu_{m,i} \mu_{m',i} \nonumber \\
    &\quad - \left( \sum_{i} \Prob_i \mu_{m,i} \right) \cdot \left(\sum_{i} \Prob_i \mu_{m',i} \right), \nonumber
\end{align}
represents the covariance between the components $\Obs_m$ and can be either positive or negative in general.

Substituting this relation into the single-shot QMETTS variance gives
\begin{align}
    \Var[\bar O](1) 
    &= \tau(1) \left(1+\frac{1}{M}\right) \Bigl[ \Var_\Gibbs[\bar O_\Gibbs] \nonumber \\
    &\quad +  \frac{2\sum_{m< m'} \Cov(\mu_{m, i}, \mu_{m', i})}{N_\tot / M}\Bigr].\nonumber 
\end{align}
Therefore, apart from the covariance correction among different components
of the observable in the METTS ensemble, the single-shot QMETTS variance is
related to the direct Gibbs measurement variance by the multiplicative factor
$\tau(1)(1+1/M)$.

\section{MUPB proof for the concrete setup} \label{app:proof}
We now provide a rigorous proof of our main statement for the concrete setup considered in the numerical simulation. 
The improved measurement bases shown in Fig.~\ref{fig:MUPB} were obtained heuristically. 
In what follows, we show, using the stabilizer formalism again, that these bases satisfy the defining properties of MUPB. 
We begin by showing that the physical $Z$-basis and $X$-basis represented in Fig.~\ref{fig:MUPB} are eigenstates of the Gauss's law operators (Item 1 of Def.~\ref{def:MUPB}).
The physical $Z$-basis is given by 
\begin{align}
    \ket{e_i^{(Z_\phys)}} = (\prod_{m=1}^{L_{\text{KS}} - 1} \sigma_{m,m+1}^\mathsf{H})^{\dagger}\ket{b_i} = : U_\mathsf{H}^\dagger \ket{b_i}. \nonumber
\end{align}
Since 
$
    G_n \ket{e_i^{(Z_\phys)}}
    = G_n U_\mathsf{H}^\dagger \ket{b_i}
    = U_\mathsf{H}^{\dagger}\cdot ( U_\mathsf{H} G_n U_\mathsf{H}^{\dagger} ) \ket{b_i}
$
and $\sigma^\mathsf{H}_{m-1, m} \sigma^X_{m-1, m} \sigma^\mathsf{H}_{m-1, m} = \sigma_{m-1, m}^Z$, we obtain
\begin{align}
    U_\mathsf{H} G_nU_\mathsf{H}^{\dagger} 
    = (-1)^n \sigma_{n-1, n}^Z Z_n \sigma_{n,n+1}^Z. \nonumber
\end{align}

As $(-1)^n \sigma_{n-1, n}^Z Z_n \sigma_{n, n+1}^Z\ket{b_i} = z \ket{b_i}$, where $z\in \{\pm 1\}$ is the corresponding eigenvalue,  we can show
\begin{align}
    &U_\mathsf{H}^{\dagger} \cdot \left( U_\mathsf{H} G_n U_\mathsf{H}^{\dagger} \right) \ket{b_i} 
    = z U_\mathsf{H}^{\dagger} \ket{b_i}
    = z \ket{e_i^{(Z_\phys)}}. \nonumber
\end{align}
As for $\{\ket{e_i^{(X_\phys)}}\}$ represented in Fig.~\ref{fig:MUPB}, since it can be written as
$
    \ket{e_i^{(X_\phys)}}
    = U^{\dagger}\ket{e_i^{(Z_\phys)}}
$ where $U$ is defined in Eq.~(\ref{eq:physX_U}),
\begin{align}
    G_n \ket{e_i^{(X_\phys)}} 
    &= G_n U^{\dagger} \ket{e_i^{(Z_\phys)}} \nonumber \\
    &= U^{\dagger} (U G_n U^{\dagger}) \ket{e_i^{(Z_\phys)}}. \nonumber
\end{align}
Then,
\begin{align}
    &G_n = (-1)^n \sigma_{n-1, n}^X Z_n \sigma_{n,n+1}^X \nonumber \\
    &\xrightarrow{V}
    (-1)^n  X_n \nonumber \\
    &\xrightarrow{\prod_m \sigma^\mathsf{H}_{m,m+1}} 
    (-1)^n  X_n \nonumber \\
    &\xrightarrow{ V^\dagger}
    (-1)^n \sigma_{n-1, n}^X Z_n \sigma_{n,n+1}^X = G_n, \nonumber
\end{align}
where $A \xrightarrow{X} B$ denotes the conjugation of $A$ by $X$, $B=XAX^\dagger$ and $V$ is defined in Eq.~(\ref{eq:physX_V}).
Thus, the Gauss's law operators are invariant under the conjugation $UG_nU^\dagger = G_n$ and we can also prove that $\{ \ket{e_i^{(X_\phys)}} \}$ are eigenstates of $G_n$.

Next, we prove the mutual unbiasedness in physical subspace: $|\braket{e_i^{(Z_\phys)}|e_j^{(X_\phys)}}| = {1}/{\sqrt{d_{\phys}}}$~(Item 2 in Def.~\ref{def:MUPB}).
The stabilizer group that stabilizes $\ket{e_i^{(Z_\phys)}}$ is given by,
\begin{align}
    &\langle \{G_n\}_{n=1}^{L_{\text{KS}} - 1} \nonumber 
    ~\cup~\{(-1)^{b_i({g_{n, n+1}})} \sigma^X_{n, n+1}\}_{n=1}^{L_{\text{KS}} - 1} \nonumber \\
    &~\cup~ \{(-1)^{b_i({f_L})}Z_{L_{\text{KS}}}\} \rangle .\nonumber
\end{align}
From Theorem~\ref{th:code_dim}  we can see that the dimension of this stabilized space is $2^0 = 1$ and the state is uniquely determined up to a global phase.
Then, the state $\ket{e^{(X_\phys)}_i}=U^{\dagger}\ket{e_i^{(Z_\phys)}}$ can be described by following stabilizer operators through the conjugation under $U$,
\begin{align}
    &\{(-1)^{b_i({g_{n, n+1}})} \sigma^X_{n, n+1}\}_n \nonumber \\
    &\xrightarrow{U} \{{(-1)^{b_i({g_{n, n+1}})} X_n\sigma^Z_{n, n+1} X_{n+1}\}_n }, \nonumber \\
    &\{Z_{L_{\text{KS}}}\} \xrightarrow{U} \{X_{L_{\text{KS}}}\}, \nonumber
\end{align}
and Gauss's law operators are invariant, $UG_n U^\dagger = G_n$.
Since different operators are $\{(-1)^{b_i({g_{n, n+1}})} \sigma^X_{n, n+1}\}_{n=1}^{L_{\text{KS}} - 1} ~\cup~ \{Z_{L_{\text{KS}}}\}$ and $\{(-1)^{b_i({g_{n, n+1}})} X_n \sigma^Z_{n, n+1} X_{n+1}\}_{n=1}^{L_{\text{KS}} - 1} ~\cup~ \{X_{L_{\text{KS}}}\}$, the number of different generators is $L_{\text{KS}}$.
Then, by Theorem~\ref{th:inner_prod}, we can compute the inner product as
\begin{align}
    |\braket{e_i^{(Z_\phys)}|e_j^{(X_\phys)}}| = \frac{1}{2^{L_{\text{KS}}/2}} = \frac{1}{\sqrt{d_{\phys}}}. \nonumber
\end{align}
Therefore the physical $Z$-basis and the physical $X$-basis satisfy the condition of the MUPB.

\section{{Implementation of QITE}} \label{app:QITE}
In this appendix, we briefly provide the details of the implementation of QITE by Motta et al.~\cite{Motta:2019yya}.
To implement the QITE via Trotterization, we decompose the Hamiltonian with the chemical potential term into three parts,
$H -\mu N = H_e + H_o + H_D$,
where
\begin{align}
H_{e/o} &= -\frac{1}{4a} \sum_{n \in \text{even/odd}}^{L_{\text{KS}}-1}
\Bigl[
X_n \sigma^Z_{n,n+1} X_{n+1} \nonumber \\
&~~+ Y_n \sigma^Z_{n,n+1} Y_{n+1}
\Bigr],\nonumber \\
H_D &= a g^2 \sum_{n=1}^{L_{\text{KS}}-1} \left( 1 - \sigma^X_{n,n+1} \right)
+ \frac{m}{2} \sum_{n=1}^{L_{\text{KS}}} (-1)^{n\rev{+1}} Z_n \nonumber \\
&~~- \frac{\mu}{2} \sum_{n=1}^{L_{\text{KS}}} Z_n. \nonumber
\end{align}
As discussed below, this decomposition suppresses Trotter errors.
We employ a second-order Trotter decomposition,
\begin{align}
e^{-\beta (H-\mu N)/2}
&= \left( e^{-\Delta \beta (H-\mu N)/2} \right)^{N_s} \nonumber \\
&\approx
\left(
\prod_{\gamma \in \{e,o,D\}}^{\leftarrow}
\prod_{\gamma \in \{e,o,D\}}^{\rightarrow}
e^{-\Delta \beta H_\gamma/4}
\right)^{N_s}, \nonumber
\end{align}
where $\Delta \beta = \beta / N_s$ and $N_s$ is the number of Trotter steps.
The left- and right-ordered products are defined as
\begin{align}
\prod_{\gamma \in \{e,o,D\}}^{\leftarrow}
e^{-\Delta \beta H_\gamma/4}
&:= e^{-\Delta \beta H_e/4}
   e^{-\Delta \beta H_o/4}
   e^{-\Delta \beta H_D/4}, \nonumber \\
\prod_{\gamma \in \{e,o,D\}}^{\rightarrow}
e^{-\Delta \beta H_\gamma/4}
&:= e^{-\Delta \beta H_D/4}
   e^{-\Delta \beta H_o/4}
   e^{-\Delta \beta H_e/4}. \nonumber
\end{align}

Following Ref.~\cite{childs2021theory},
the Trotter error can be bounded in the spectral norm as
\begin{align}
&\Bigg\|
\prod_{\gamma}^{\leftarrow}
\prod_{\gamma}^{\rightarrow}
e^{-\Delta \beta H_\gamma/4}
- e^{-\Delta \beta \rev{(H - \mu N)/2}}
\Bigg\|_{\infty} \nonumber \\
&=
\mathcal{O}
\left(
\tilde{\alpha}_{\mathrm{comm}}
(\Delta \beta)^3
e^{4 \Delta \beta \sum_\gamma \|H_\gamma\|_\infty}
\right), \nonumber
\end{align}
where
$\tilde{\alpha}_{\mathrm{comm}}
= \sum_{\gamma_1,\gamma_2,\gamma_3}
\|[H_{\gamma_1},[H_{\gamma_2},H_{\gamma_3}]]\|_\infty$. 
Thus, $\Delta \beta \rightarrow 0$ recovers the exact imaginary-time evolution operators.
In the present work we take $\Delta \beta =0.01$.



\section{{Dependence of the probability distributions on the initial state in the numerical benchmark}} \label{app:init_depend}
{In this appendix, we show the sampling distributions obtained from several different initial states using the MUPBs for the numerical benchmark discussed in Sec.~\ref{sec:numerical}.
We use $N_\chain = 1000$ for each initial state, matching the setup in Sec.~\ref{sec:numerical}.
As seen in Fig.~\ref{fig:init_depend}, the sampled collapse-state distributions obtained with the MUPBs are all consistent with the exact stationary distribution within finite-sampling fluctuations, indicating that no clear dependence on the initial state is observed in the present setup.}
\begin{figure*}
    \centering
    \includegraphics[width=0.9\linewidth]{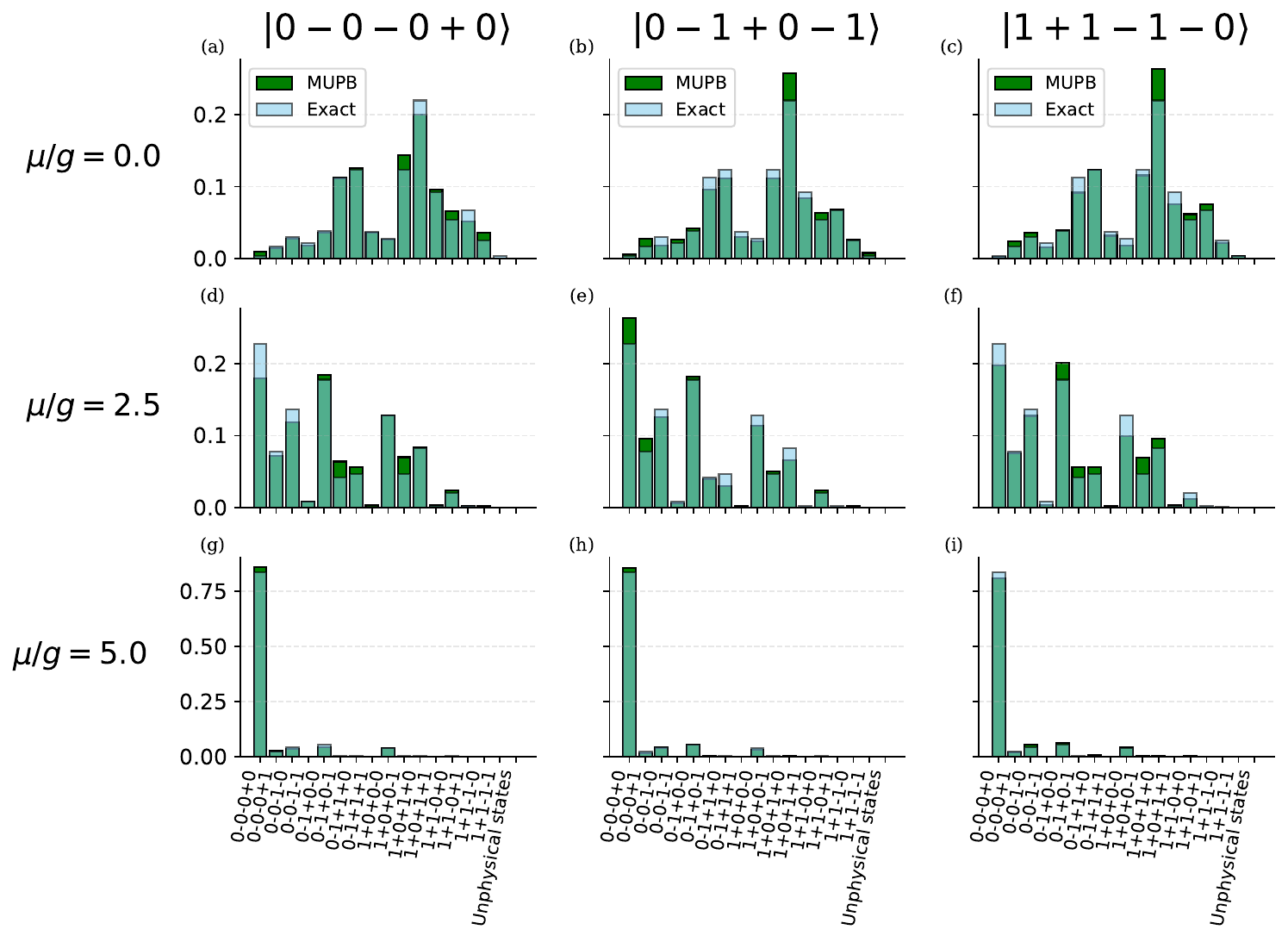}
    \caption{
    Dependence of probability distributions on the initial state.
    Sampling distributions of collapse states obtained from projective measurements using the MUPBs during QMETTS sampling (green) and the exact probabilities (light blue) for $L_{\text{KS}}=4$.
    \rev{The ITE is implemented exactly using classical computation.}
    The panels are arranged by chemical potential $\mu/g$ (rows) and by the initial state: $\ket{0-0-0+0}$ (left), $\ket{0-1+0-1}$ (middle), and $\ket{1+1-1-0}$ (right).
    The inverse temperature is fixed at $\beta g=1.0$.
    Top panels (a--c) correspond to $\mu/g = 0.0$, middle panels (d--f) to $\mu/g = 2.5$, and bottom panels (g--i) to $\mu/g = 5.0$.
    The exact probabilities are calculated as $\Prob_i = \langle i^{(Z_\phys)}| e^{-\beta(H-\mu N)} |i^{(Z_\phys)} \rangle / Z$.
    All QMETTS results are obtained with $N_\chain = 1000$.
    }
    \label{fig:init_depend}
\end{figure*}

\section{\rev{Autocorrelation functions}} \label{app:auto_func}
\rev{In this appendix, we examine the behavior of the autocorrelation function $\rho_{\mathcal{O}}(t,N_{\mathrm{shot}}=1)$ defined in Eq.~(\ref{eq:auto_corr_w_shot}) for the three observables $\mathcal{O}=h$, $\bar{\psi}\psi$, and $n$.
Figure~\ref{fig:auto_correlation_function} shows the corresponding autocorrelation functions. 
The shaded ``95\% band'' indicates the approximate 95\% bounds around zero expected under the assumption of uncorrelated samples.
For all cases, the correlations decay rapidly with increasing lag time $t$ and enter the shaded regions within a few ($\sim 3$) steps. 
Based on this behavior, we choose the window size $w=3$ for estimating the integrated autocorrelation time.}
\begin{figure*}
    \centering
    \includegraphics[width=0.9\linewidth]{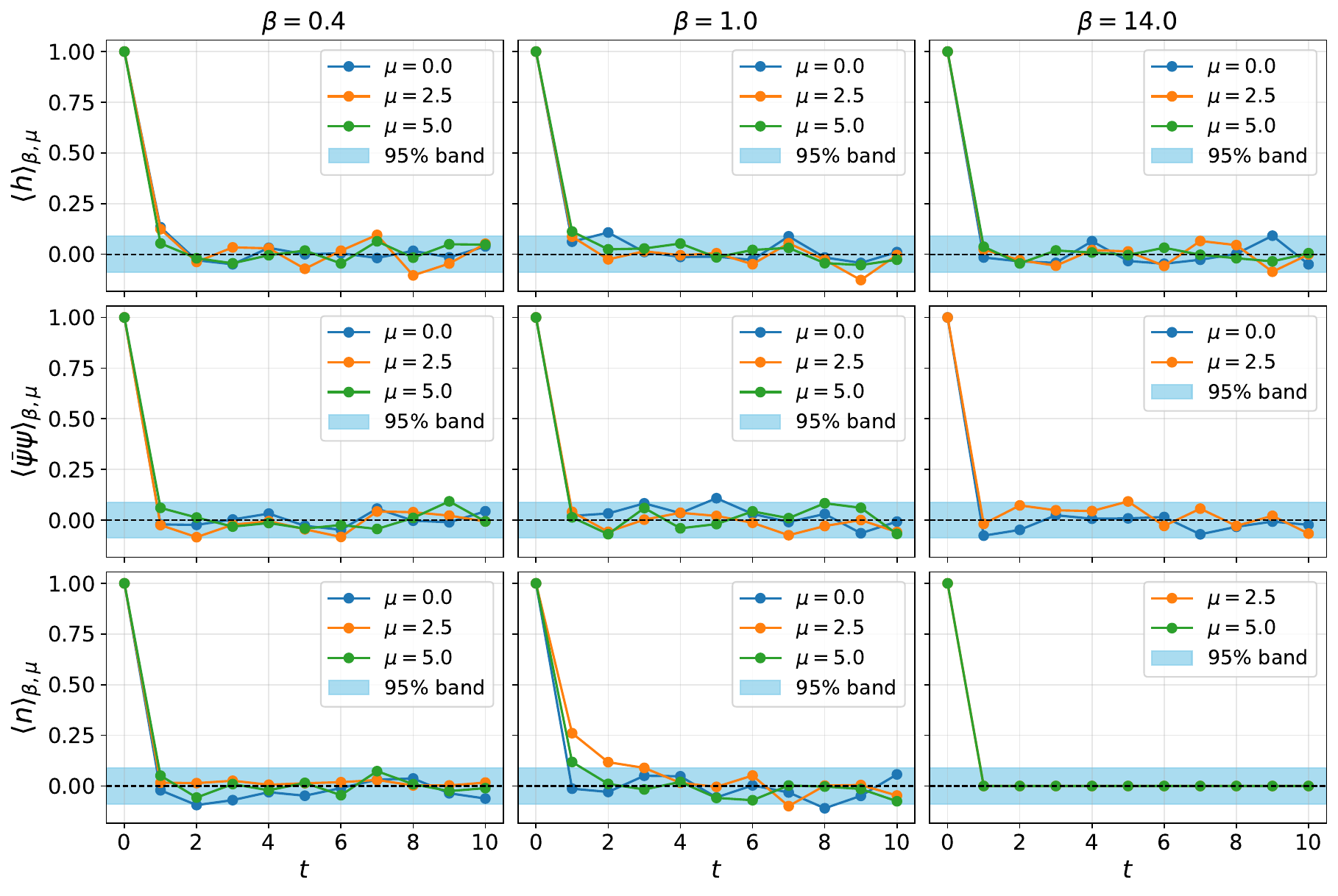}
    \caption{
    \rev{Autocorrelation functions $\rho_{\mathcal{O}}(t, N_{\shot}=1)$ for the energy density (top row), chiral condensate (middle row), and fermion number density (bottom row).
    The inverse temperature is fixed at $\beta g=0.4$ (left column), $\beta g=1.0$ (middle column), and $\beta g=14.0$ (right column).
    Blue, orange, and green markers correspond to chemical potentials $\mu/g=0.0$, $2.5$, and $5.0$, respectively.
    The shaded ``95\% band'' indicates the approximate 95\% confidence interval around zero expected under the assumption of uncorrelated samples.
    For the cases $\mathcal{O} = {\bar \psi \psi}$, $\beta g = 14.0$, and $\mu /g=5.0$ and $\mathcal{O} = {n}$, $\beta g = 14.0$, and $\mu /g=0.0$, we obtain zero variance and cannot define the corresponding autocorrelation functions.}}
    \label{fig:auto_correlation_function}
\end{figure*}

\section{\rev{System size dependence of the integrated autocorrelation time}} \label{app:system_size_depend_autocorr}
\rev{In this appendix, we investigate the system-size dependence of the integrated autocorrelation time by varying $L_{\text{KS}}=2, 4, 6,$ and $8$.
We use the window parameter $w$ introduced in Sec.~\ref{sec:numerical} and set $w=3$ based on the analysis of autocorrelation functions in App.~\ref{app:auto_func}.
Figure~\ref{fig:auto_correlation_system_size_dependence_exactITE} shows the resulting system-size dependence of the integrated autocorrelation times $\tau(1)$.
Within the range of system sizes considered, we observe no clear systematic dependence of the integrated autocorrelation time on the system size. In all cases, the estimated autocorrelation times remain of order unity.}
\begin{figure*}
    \centering
    \includegraphics[width=0.95\linewidth]{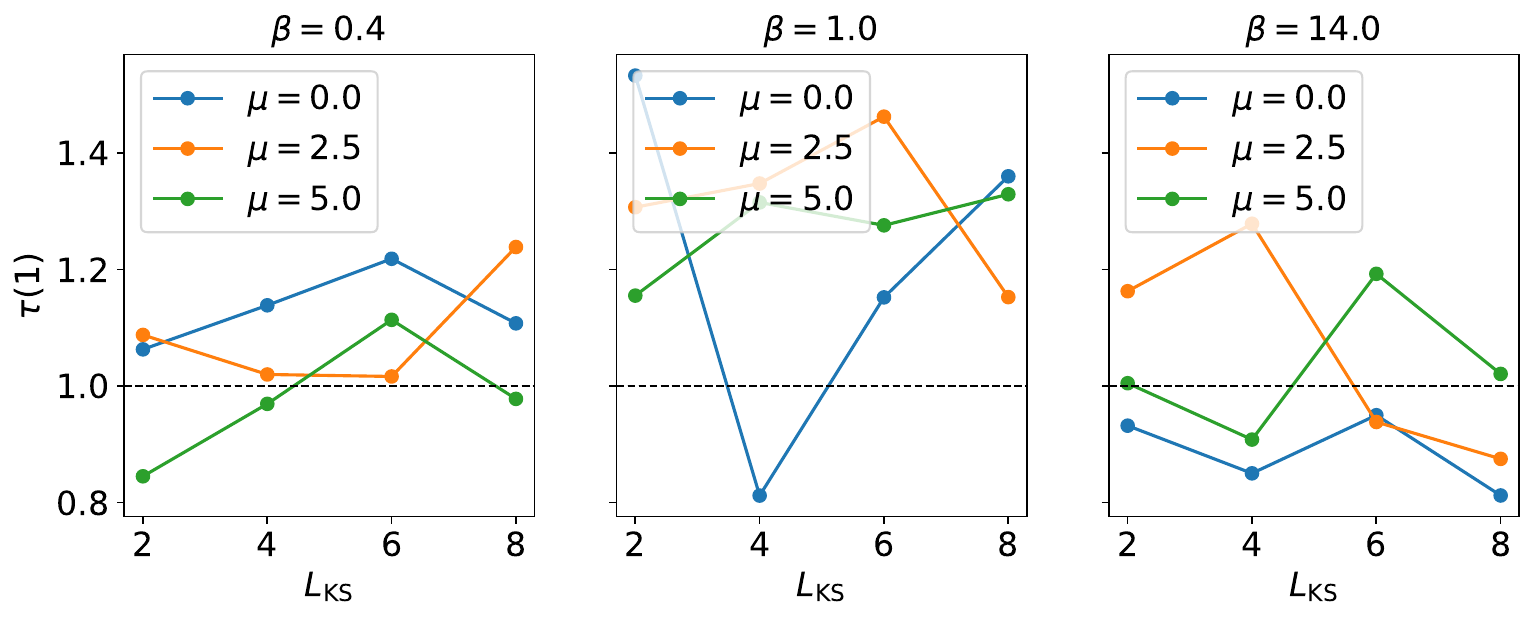}
    \caption{
    \rev{System-size dependence of the integrated autocorrelation time for the energy density. 
    The left, middle, and right panels correspond to $\beta g=0.4$, $1.0$, and $14.0$, respectively. In each panel, blue, orange, and green markers correspond to $\mu/g=0.0$, $2.5$, and $5.0$, respectively.}
    }
    \label{fig:auto_correlation_system_size_dependence_exactITE}
\end{figure*}

\bibliographystyle{utphys}
\bibliography{aps}

\end{document}